\documentclass[12pt]{article}
\usepackage{authblk}
\usepackage{geometry}
\usepackage{amsmath}
\usepackage{amssymb}
\usepackage{amsthm}
\usepackage{mathrsfs}
\usepackage{subfigure}
\usepackage{customcommands}
\usepackage{bm}
\usepackage{Nettastyle}

\usepackage[normalem]{ulem}

\preprint{MIT-CTP/5265}

\title{Holography Abhors Visible Trapped Surfaces}

\author{Netta Engelhardt}
\author{and \AA{}smund Folkestad}
\emailAdd{engeln@mit.edu}
\emailAdd{afolkest@mit.edu}
\affiliation{Center for Theoretical Physics, Massachusetts Institute of Technology, \\Cambridge, MA 02139, USA}

\abstract{We prove that consistency of the holographic dictionary implies a hallmark prediction of the weak cosmic censorship
    conjecture: that in classical gravity, trapped surfaces lie behind event horizons. In particular, the existence of a trapped surface
    implies the existence of an event horizon, and that furthermore this event horizon must be outside of the trapped
    surface. More precisely, we show that the formation of event horizons outside of a strong gravity region is a direct consequence of causal wedge inclusion, which is required by entanglement wedge reconstruction.   We make few  assumptions beyond the absence of evaporating
    singularities in strictly classical gravity. We comment on the implication that spacetimes with naked trapped surfaces do not admit a holographic dual, note a possible application to holographic complexity, and speculate on the dual CFT interpretation of a trapped surface.} 

\begin{document}

\maketitle

\section{\label{sec:intro}Introduction}
Foundational results in modern gravitational physics, from black hole thermodynamics~\cite{Haw71, Bek72, BarCar73} to topological censorship~\cite{FriSch93}, often rely on the absence of strong gravity outside of horizons, or more precisely the weak cosmic censorship conjecture~\cite{Pen69}. Gedankenexperiments that rely upon the formation of black holes from generic matter collapse require an implicit assumption that matter typically coalesces into black holes and not into naked singularities. 

This assumption appears reasonable given the apparent absence of naked singularities in observational data to date.
The more precise statement of weak cosmic censorship (which we shall henceforth refer to simply as
``cosmic censorship'') is the conjecture that the Einstein equation evolves generic regular initial data with certain
asymptotics to a complete asymptotic infinity $\mathscr{I}$~\cite{GerHor79, Wal97}. In its initial formulation it applied
specifically to asymptotically flat initial data in four dimensions~\cite{Pen69}, but it has been generalized to other dimensions and asymptotics when a suitable
asymptotic infinity $\mathscr{I}$ exists (see e.g.~\cite{SantosTalk}). 

Developments over the past decade have eroded confidence in the validity of this statement at least in its more general form: counterexamples have been found in asymptotically AdS$_{4}$~\cite{HorSan16,CriSan16, CriHor17, CriHor18, HorSan19} and in higher
dimensional asymptotically flat space~\cite{GreLaf93,GreLaf94,LehPre10,FigKun15,FigKun17,AndEmp18,AndEmp19,AndFig20}. Even in four-dimensional asymptotically flat space, non-generic initial data can evolve to form naked singularities~\cite{Cho92,Chr94,Ham95}. While it is in principle possible that cosmic censorship is in fact correct in four-dimensional asymptotically flat space (and for generic initial data), numerous extant violations in other settings suggest otherwise. Since results that rely on cosmic censorship are expected to be applicable in broad generality in arbitrary dimensions and often for AdS and dS asymptotics, violations of cosmic censorship -- generic or otherwise -- are problematic in any setting where classical gravity is expected to be valid.

It would be particularly unfortunate\footnote{Or fortunate, from a certain point of view.} if such violations were to indicate that trapped surfaces can lie outside of event horizons. In its current formulation, cosmic censorship forbids the existence of trapped surfaces -- i.e. surfaces from which light rays converge in any direction due to gravitational lensing -- outside of event horizons~\cite{HawEll}. A large set of theorems in classical gravity relies on this result (see~\cite{Wald, HawEll}). Numerical relativity typically uses the detection of (marginally) trapped surfaces as an avatar for the event horizon, whose location (and existence) can only be determined in infinite time.  Must we face the possibility that these results are all questionable?

Possibly not, at least in spacetimes that arise as classical limits of quantum gravity.  The recent discoveries of violations of cosmic censorship in AdS$_{4}$ \cite{HorSan16,CriSan16, CriHor17, CriHor18, HorSan19} have also been found to violate the Weak Gravity Conjecture~\cite{ArkMot06}; cosmic censorship is restored precisely when the theory is adjusted so as to satisfy this conjecture, which is hypothesized to discriminate between spacetimes that do and do not admit valid UV completions. The confluence of validity of the cosmic censorship and the weak gravity conjecture has given rise to speculation that while classical General Relativity admits violations of cosmic censorship, the classical spacetimes that result from a truncation of a valid quantum theory of gravity do not~\cite{CriHor17,HorSan19}: that is, that quantum gravity enforces cosmic censorship on its strict classical limit. 

Here we focus on trapped surfaces rather than the statement of cosmic censorship in terms of initial data. The
restriction of trapped surfaces to lie behind horizons is one of the most valuable consequences of cosmic censorship, since as noted above,
it is a sine qua non for a number of results in General Relativity~\cite{Wald}. It is furthermore one of relatively few
consequences of cosmic censorship that can be formulated in terms of the experience of a family of observers: a trapped
surface outside of a horizon can in principle be detected by an asymptotic family of observers in finite time
(or retarded time, in the asymptotically flat case). In fact, one could even go so far as to argue that the absence of trapped surfaces
outside of horizons is a large part of the physical content of the weak cosmic censorship conjecture; namely that regions of strong gravity (usually heralded by trapped surfaces) are 
hidden from asymptotic observers. This statement also fortunately avoids any references to singularities, which are
notoriously hard to work with -- see~\cite{Tip77,Cla84, Dra84,Nol99,Ori00} and references therein for literature in
General Relativity attempting classify strengths and types of singularities\footnote{For instance, as discussed in~\cite{Emp20}, a singularity like the Gregory-Laflamme instability would by any nice definition be considered ``weak'' enough to be allowed by cosmic censorship. We thank R. Emparan for discussions on this topic.}.  

We aim to test the hypothesis that quantum gravity forces trapped surfaces behind horizons: we use holography as a laboratory for the classical limit of quantum gravity and ask whether some principle of the AdS/CFT correspondence implies that trapped surfaces remain cloaked from the asymptotic boundary. 

Concrete evidence in favor of such a conclusion was found
recently in e.g.~\cite{EngHor19}, which used the holographic dictionary to prove the Penrose Inequality in
AdS~\cite{Pen73,ItkOz11,HusSin17}, a
key result implied by the combination of two oft-quoted but unproven conjectures: (1) that trapped surfaces lie behind horizons, and (2) that black
holes equilibrate. The proof of~\cite{EngHor19} assumed neither (1) nor (2) but instead made use of the holographic
entanglement entropy proposal of Ryu-Takayanagi~\cite{RyuTak06} and Hubeny-Rangamani-Takayanagi~\cite{HubRan07} (HRT)
\begin{equation}\label{eq:HRT}
S_{\mathrm{vN}}[\rho_{R}]=\frac{\mathrm{Area}[X_{R}]}{4G\hbar},
\end{equation}
where $\rho_{R}$ is the density matrix of the CFT state reduced to the region $R$ and $X_{R}$ is the minimal area stationary surface homologous to $R$. 

Because the Penrose Inequality follows from the absence of trapped surfaces outside of horizons together with black
hole equilibration, it is a good omen in
favor of cosmic censorship; however, it falls well short of proving that trapped surfaces in fact must lie behind horizons\footnote{And falls even shorter of proving cosmic censorship.}.

In this article, we close this gap, thus proving a central consequence of cosmic censorship: the holographic dictionary implies that
trapped surfaces lie behind 
event horizons. Our primary assumptions are (1) the HRT prescription,  (2) that there exist unitary operators on the boundary whose effect in the
bulk propagates causally, and (3) that singularities do not evaporate (a criterion that will be defined more rigorously in the following section) in classical gravity without violations of the null energy condition. Steps (1) and (2) together imply causal wedge inclusion~\cite{Wal12, HeaHub14}: that the causal wedge must lie inside of the entanglement wedge. This holographic ingredient is a key step in our proof.

Before proceeding to the outline of the proof, let us briefly comment on quantum corrections. When quantum gravity effects are taken into account and violations of the null energy
condition ($T_{ab}k^{a}k^{b}\geq 0$ for null $k^{a}$) are permitted, there is of course no expectation that the above formulation of cosmic censorship should remain valid; it is, after all, a prediction about the behavior of classical general relativity. Indeed, evaporating black holes do  have trapped
surfaces outside of horizons (see e.g.~\cite{BouEng15c}). However, cosmic censorship is a statement about the classical theory; we are primarily concerned with ascertaining whether the classical limit of quantum gravity features trapped surfaces outside of event horizons. Within semiclassical gravity, however, these results may be interpreted as statements about early stages of gravitational collapse.

\paragraph{Outline of the Proof:}  the proof has three ingredients: first, the HRT prescription~\eqref{eq:HRT} combined
with the fact that turning on local unitary operators in the dual theory at $\mathscr{I}$ results in causally propagating bulk perturbations; second, a
theorem of~\cite{AndMet07} proving that an apparent horizon must lie between trapped surfaces and ``normal'' surfaces --
surfaces in which ingoing light rays converge and outgoing light rays expand; and third,
the holographic description of apparent horizons~\cite{EngWal17b, EngWal18} (which is not an additional ingredient but rather a
construction that relies only on HRT).

First, we use the fact that $S_{\mathrm{vN}}[\rho]$ is invariant under local unitary operations on the boundary state $\rho$ to argue
that the HRT surface $X$ of a connected component of $\mathscr{I}$ cannot be timelike separated to any portion of
$\mathscr{I}$. This is a well-established requirement often referred to as {causal wedge inclusion~\cite{Wal12, HeaHub14}: if $X$ were timelike-separated to $\mathscr{I}$, then it would be
possible for a local unitary CFT operator acting to create a bulk signal propagating causally to $X$, which could modify its area; this is
illustrated in Fig.~\ref{fig:usignal}. Thus we would find that if $X$ were timelike to any $p\in\mathscr{I}$ it would be possible to modify $S_{\mathrm{vN}}[\rho]$ via local unitaries acting on $\rho$.\footnote{
Causal wedge inclusion can also be shown to follow from entanglement wedge nesting \cite{AkeKoe16}, which states that if
the size of a CFT subregion increases from $R_1$ to $R_2 \supset R_1 $, the part of the bulk that could be reconstructed
from $R_1$ can also be reconstructed from $R_2$.}

\begin{figure}
\centering
\includegraphics[width=0.4\textwidth]{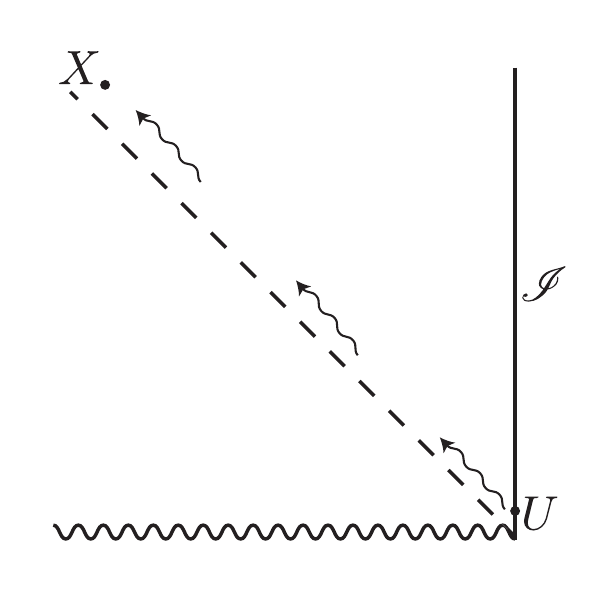}
\caption{Example showing how causal connection between $\mathscr{I}$ and its HRT surface $X$ can be used to change
$S_{\mathrm{vN}}$ with a local unitary.}
\label{fig:usignal}
\end{figure}

Next, the theorem of~\cite{AndMet07} guarantees that a certain type of apparent horizon always lies between normal and
trapped surfaces. We show that, given a trapped surface in our setup, there are always normal surfaces outside of it,
and that the type of apparent horizon guaranteed by the theorem satisfies the refinement necessary for step three.

Finally, we use the holographic construction of a dual to this type of apparent horizon~\cite{EngWal17b, EngWal18}. The construction
instructs us to fix the spacetime and matter outside of the apparent horizon -- its so-called outer wedge -- and modify
the spacetime elsewhere via a specific prescription. If naked
singularities and Cauchy horizons develop as a result of this prescription, we permit any extension beyond the Cauchy
horizon consistent with our relatively mild assumptions about General Relativity (here we operate under the assumption that the boundary conditions are inherited by a top-down UV completion). In this newly constructed spacetime,
the HRT surface $X$ is null-separated from the apparent horizon as illustrated in Fig.~\ref{fig:Xandmu}.
Since by step one there are no timelike curves from the HRT surface to $\mathscr{I}$, there can be no future-directed timelike curves
from the apparent horizon to $\mathscr{I}$: we find that apparent horizons must lie outside of $I^{-}[\mathscr{I}]$ in
the coarse-grained spacetime; under the assumption that singularities do not evaporate in classical spacetimes
satisfying the null energy condition, we can then deduce that apparent horizons must therefore also lie behind the event
horizon in the original spacetime.\footnote{Since, to our knowledge, known solutions with evaporating singularities -- a concept that we
will make precise in Section~\ref{sec:StepOne} -- arising from the evolution of initial data feature violations of the null energy condition, we do not find this assumption to be particularly prohibitive. More generally, violations of weak cosmic censorship are normally concerned with the formation of singularities rather than their demise. This will be discussed at greater length in Section~\ref{sec:disc} .} This immediately shows that naked singularities -- should they exist in holographic spacetimes -- cannot incur trapped surfaces outside of horizons. The contrapositive then yields the following statement: if a given spacetime has a trapped surface outside of (or without a) horizon, it cannot be holographic; should we adopt the perspective that AdS$_{D\geq 3}$ spacetimes with a well-defined UV completion are always holographic, this then becomes a potential swampland condition on the set of spacetimes with valid UV completions.

\begin{figure}
    \centering
    \subfigure[]{\includegraphics[height=5cm]{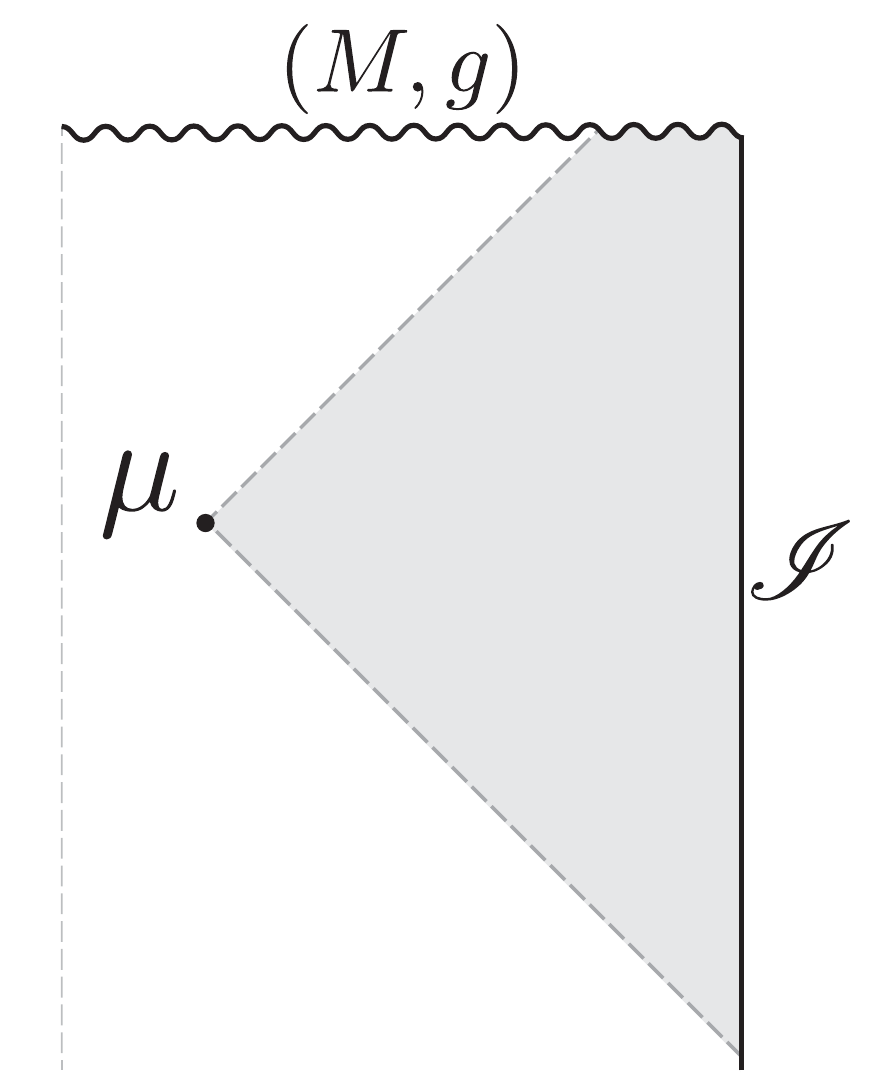} \label{xmurel1}}
    \subfigure[]{\includegraphics[height=5cm]{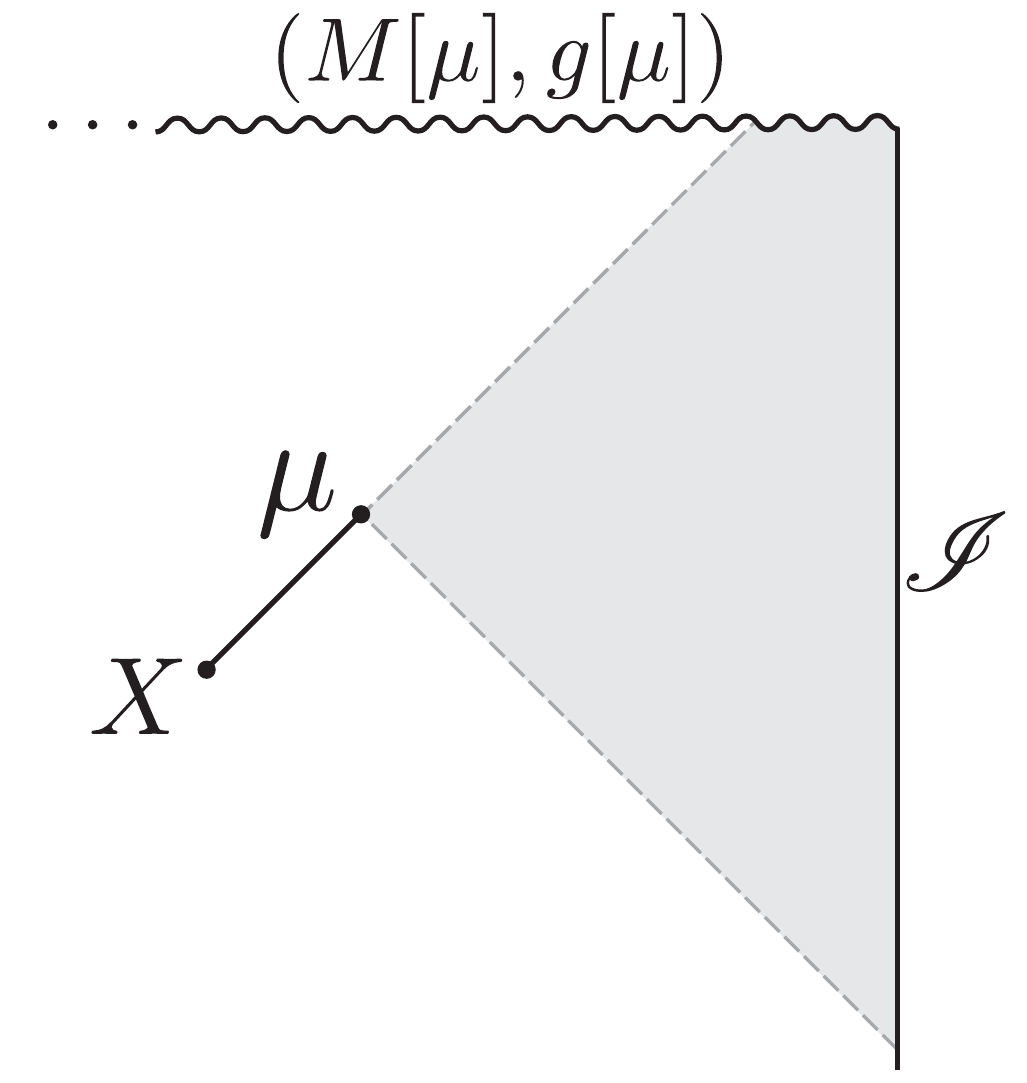} \label{xmurel2}}
    \caption{\subref{xmurel1} An apparent horizon $\mu$ (more precisely, a minimar surface) in a spacetime $(M, g)$. \subref{xmurel2} The same apparent horizon in its corresponding
    coarse-grained spacetime $(M[\mu], g[\mu])$, where there HRT surface $X$ with respect to $\mathscr{I}$ is null
    related to $\mu$. The shaded wedge outside $\mu$ is common to both spacetimes.}
\label{fig:Xandmu}
\end{figure}

\paragraph{Relation to Prior Work:} As alluded to earlier, previous attempts to prove various implications of the weak cosmic censorship
conjecture have frequently encountered complications related to classifying singularities~\cite{Tip77,Cla84, Dra84,Nol99,Ori00}. By considering trapped surfaces directly, and by defining evaporating singularities in terms of
homology hypersurfaces (see Definition~\ref{def:evapsing}), we are able to avoid this complication altogether. Additionally, since cosmic censorship is known to be
violated on a measure-zero set of the space of solutions to the Einstein equations \cite{Cho92,Chr94,Ham95}, a genericity assumption is
necessary for there to be any chance of the usual statement weak cosmic censorship to be true. The first half of our proof requires no assumptions about any genericity condition, and proves definitively
that spacetimes with marginally trapped surfaces outside of horizons are not holographic. The second part of the proof,
that all trapped surfaces are also behind horizons does assume a mild genericity condition and a technical assumption. We
expect that these assumptions likely can be relaxed. 

\subsection{Assumptions and  Conventions}
\begin{figure}
\centering
\includegraphics[width=0.4\textwidth]{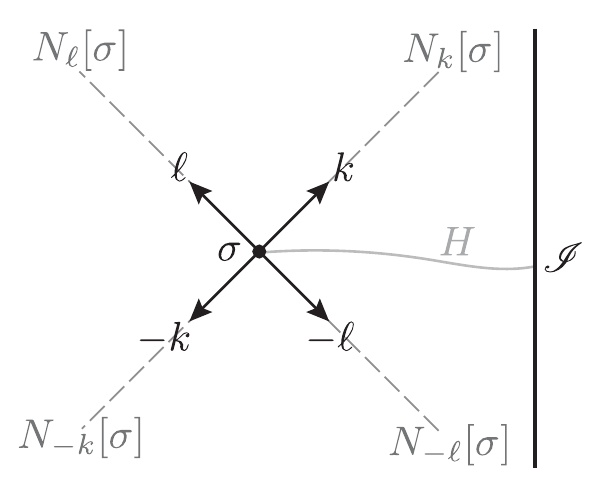}
\caption{An illustration of the four null normals of a surface $\sigma$ (here shown to be homologous to $\mathscr{I}$) together with the corresponding null congruences. }
\label{fig:tangents}
\end{figure}

\paragraph{Assumptions:} We will assume the null convergence condition, $R_{ab}k^{a}k^{b}\geq 0$ for all null vectors
$k^{a}$ as well as the HRT prescription for computing $S_{\mathrm{vN}}$~\eqref{eq:HRT} and entanglement wedge reconstruction. In particular, this means that we work strictly within classical General Relativity. We further assume that our spacetime
is time-orientable and has no closed timelike curves. We further assume that in a maximal conformal
extension $(\partial M, h)$ of $\mathscr{I}$, $\mathscr{I}$ is both globally hyperbolic and geodesically complete;
unless explicitly noted otherwise, by $\mathscr{I}$ we will always mean a connected component of $\partial M$. By
contrast, $\partial M$ will always denote the complete conformal boundary. We assume that $\mathscr{I}$ is spatially compact. 

\paragraph{Conventions:} 
In describing domains of dependence that include the relevant
portion of the asymptotic boundary, we will include $\mathscr{I}$ in any Cauchy slices of these domains of dependence,
per standard conventions (see~\cite{EngWal14} for a definition).  

\begin{enumerate} 
    \item Bulk: we assume that there are no closed causal curves and
        that $J^+(p)\cap J^-(q)$ is compact in the conformal completion for each pair of points $q, p$. This is typically referred to as the AdS version of global hyperbolicity; in a slight abuse of notation, we shall simply refer to it as global hyperbolicity. We refer to the past Cauchy horizon of $D$ as $\partial^{-}D$ and the future Cauchy horizon as $\partial^{+}D$. 
For a set $A \subset M$, the boundary of $A$ is denoted $\partial A$. We take $\hat{A}$ to be the closure of $A$ in the
conformal completion $M\cup \partial M$. If an intersection between a subset of $M$ and $\mathscr{I}$ is taken, it is
always implicitly assumed that the closure in the conformal completion is taken first. 
Examples of the sets described above are shown in Fig.~\ref{fig:boundarydefs}.  

By a surface, we will always mean a spacelike (achronal) codimension-two embedded submanifold (without boundary). A surface $\sigma$ is said to be homologous to $\mathscr{I}$ if there is an achronal \textit{homology}
        hypersurface $H$ between $\sigma$ and $C$; i.e. if there is an achronal hypersurface $H$ such
        that (1) $\partial H=\sigma \cup C$ where $C$ is a Cauchy slice of $\mathscr{I}$ and (2) $H$ is compact in the
        conformal completion of $(M,g)$. Any surface $\sigma$ (homologous to $\mathscr{I}$ or otherwise) has two linearly independent normals, which we may pick to be null. We will denote
        the future-directed null normals $k^{a}$ and $\ell^{a}$; if $\sigma$ is additionally homologous to $\mathscr{I}$, then at
        least one of these null normals points towards $\mathscr{I}$ -- i.e. is ``outwards-pointing''; we will name that vector $k^{a}$.\footnote{If $\partial M$ consists of multiple connected components $\mathscr{I}$ and $\sigma$ is homologous to multiple components, it is possible for both $k^{a}$ and $\ell^{a}$ to be outwards-pointing, and it will be either clear from context or immaterial which one is which.} The four null geodesic congruences obtained by firing null geodesics along $\pm k^a$
        and $\pm \ell^a$ from $\sigma$ and terminating at conjugate points and geodesic intersections are denoted $N_{\pm k}[\sigma]$ and $N_{\pm \ell}[\sigma]$. See
        Fig.~\ref{fig:tangents} for an illustration of these concepts. 
        We will generally be interested in the expansion of a null congruence in the $n^{a}\in \{\ell^{a},k^{a}\}$ direction, defined as 
        \begin{equation}
        \theta_{n}=n^{a}\nabla_{a} \left (\ln\mathrm{Area}[\sigma] \right).
        \end{equation}
        A marginally trapped surface satisfies $\theta_{k}=0$ and $\theta_{\ell}<0$ (whereas
if $\theta_{\ell}$ is unconstrained it is simply marginally \textit{outer} trapped). A trapped surface has negative expansion
along both future expansions.

Finally, as described in the introduction, we will work with what we shall refer to \textit{permissible extensions} of a
Cauchy horizon as an extension of a spacetime beyond a Cauchy horizon that is consistent with General
        Relativity and all our global assumptions
listed above, together with one more condition:  we assume that we never extend spacetimes beyond the ``holographic region'', should
singularities terminating CFT evolution arise (see discussion below). All other conventions are as in~\cite{Wald}. 

\item Boundary: We use the letter $C$ to refer to Cauchy slices of any maximal conformal extension
    of $\mathscr{I}$, and by $i^+$ we mean future timelike infinity of $M$.\footnote{Even though the commonly drawn conformal diagram of AdS does not show $i^{+}$, it nonetheless exists!}.        
        Since the utility of our proof is in the provision of a UV complete description of the
        gravitating system via the CFT, we are also concerned with the evolution of the boundary
        theory. That is, if at any point the CFT (in the large-$N$ limit) becomes sick (e.g. if the
        stress tensor becomes divergent~\cite{CheWay19}, Hamiltonian becomes unbounded from below,
        etc.~\cite{HerHor05}) in finite boundary time in a maximal conformal extension, we must conclude
        that this similarly puts an end to bulk evolution. Thus if the  CFT evolution is
        well-defined (for $N\rightarrow \infty$) only between two (potentially empty) boundary time
slices $C^-$ and $C^+$ of the maximal conformal extension $(\partial M, h)$, $C^+$ being to the future of $C^-$, then we excise $J^{+}[C^+]$ and
$J^{-}[C^-]$ from $(M, g)$. This is because these regions are not encoded in the CFT state between $C^-$ and $C^+$, and
so the spacetime without the excision is not completely encoded in the CFT and thus not holographic.\footnote{An example
    of a scenario like this is Choptuik critical collapse in AdS
\cite{CheWay19}, where the CFT stress tensor becomes singular when the Cauchy horizon caused by the naked singularity
reaches $\mathscr{I}$. Any potential extension of the spacetime to the future of this time is not holographically encoded in the CFT.}     To facilitate terminology here, we define $\widetilde{i^+}$ as the futuremost endpoint of the CFT evolution. If the CFT evolution exists for all time in the maximal conformal extension, then $i^{+}=\widetilde{i^{+}}$; otherwise, if the CFT evolution ends prematurely at some finite time, $\widetilde{i^{+}}$ becomes the effective timelike infinity: any signals that travel from the bulk to the boundary and arrive in the future of $\widetilde{i^{+}}$ are non-holographic. 
\end{enumerate}

\begin{figure}
\centering
\includegraphics[width=0.6\textwidth]{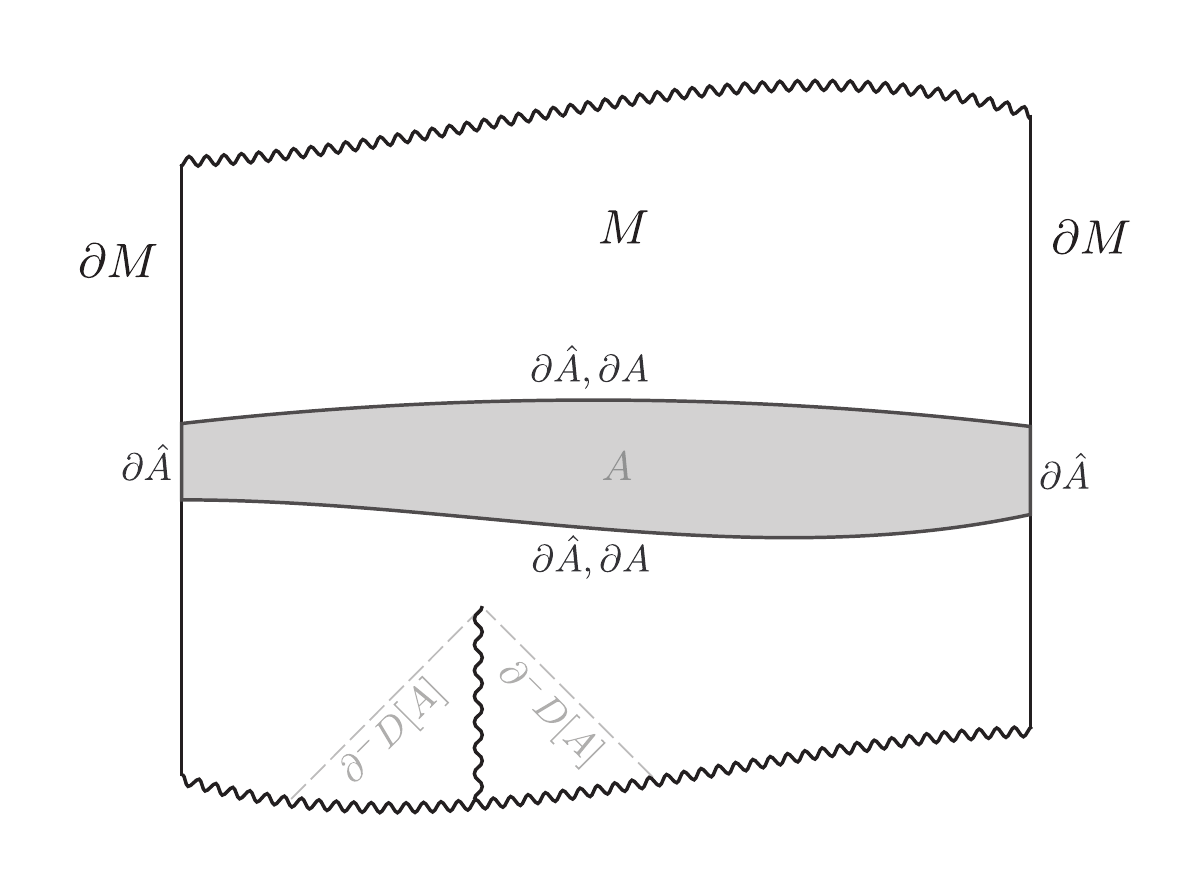}
\caption{An example of an asymptotically AdS spacetime $M$ with conformal boundary $\partial M$, and with set $A\subset M$
marked. The boundaries of $A$, $\partial A$ in $M$ and $\partial \hat{A}$ in $M\cup \partial M$, are highlighted.
Furthermore, the past Cauchy horizon of $D[A] = D[\partial A]$ is illustrated.
}
\label{fig:boundarydefs}
\end{figure}

\section{Apparent Horizons in Holography}

As discussed in Section~\ref{sec:intro}, an apparent horizon is intuitively the boundary to the very strong gravity
region on a given spatial slice; more precisely, it is the boundary between normal and trapped surfaces on a spatial
slice. In this article we will be concerned with a slight refinement of apparent horizons called  ``minimar
surfaces''~\cite{EngWal17b}.

A compact, connected surface $\mu$ is called a \textit{future minimar} surface if it satisfies the following:
\begin{enumerate}
    \item $\mu$ is marginally trapped;
    \item $\mu$ is homologous to $\mathscr{I}$;
    \item $\mu$ is \textit{strictly stable}, meaning there exists choice of $k^a$ and $\ell^a$ such that
        $k^{a}\nabla_{a}\theta_{\ell}<0$ everywhere on $\mu$ (this effectively means that there exist untrapped surfaces outside of $\mu$ and trapped surfaces inside $\mu$);
    \item $\mu$ is the surface of least area on at least one of its homology hypersurfaces.
\end{enumerate}
In analogy with the entanglement wedge, the
``outer wedge '' of $\mu$, denoted $O_{W}[\mu]$, is the domain of dependence of $H$. The absence of global hyperbolicity
raises a potential concern that the outer wedge is dependent on the choice of homology hypersurface. In the subsequent
section, we will prove that this concern is unfounded in the absence of evaporating singularities. Finally, note that it is
possible to define a past analogue of a minimar by replacing $\ell \rightarrow -k$, $k \rightarrow -\ell$. All of our results
will be valid for this time-reversed choice, but for brevity we will focus on future case. 
Consequently we often refer to future minimars just as minimars.

Our results rely heavily on the primary construction of~the dual to the apparent horizon, which builds a spacetime in
which the area of the HRT surface is identical to the area of $\mu$, and which agrees with $(M,g)$ on $O_{W}[\mu]$. The
construction is illustrated and described in Fig.~\ref{fig:coarsegrain}. 
Because the spacetime in question is generated by discarding the region that was originally behind
$\mu$, it is often referred to as the coarse-grained spacetime, denoted $(M[\mu], g[\mu])$. Here we include in
$(M[\mu],g[\mu])$ both the Cauchy development of the homology slice $H$ and any permissible extension of the Cauchy horizon. 
This means that $(M[\mu],g[\mu])$ is not necessarily globally hyperbolic, and so some additional work, carried out in
the proof of Lemma~\ref{lem:H1H2} and Theorem~\ref{thm:muhide}, is needed to show that the extremal surface $X$ on
$N_{-k}[\mu]$ about which we CPT reflect still is the HRT surface of $\mathscr{I}$ in $(M[\mu],g[\mu])$.

\begin{figure}
\centering
\includegraphics[width=1.0\textwidth]{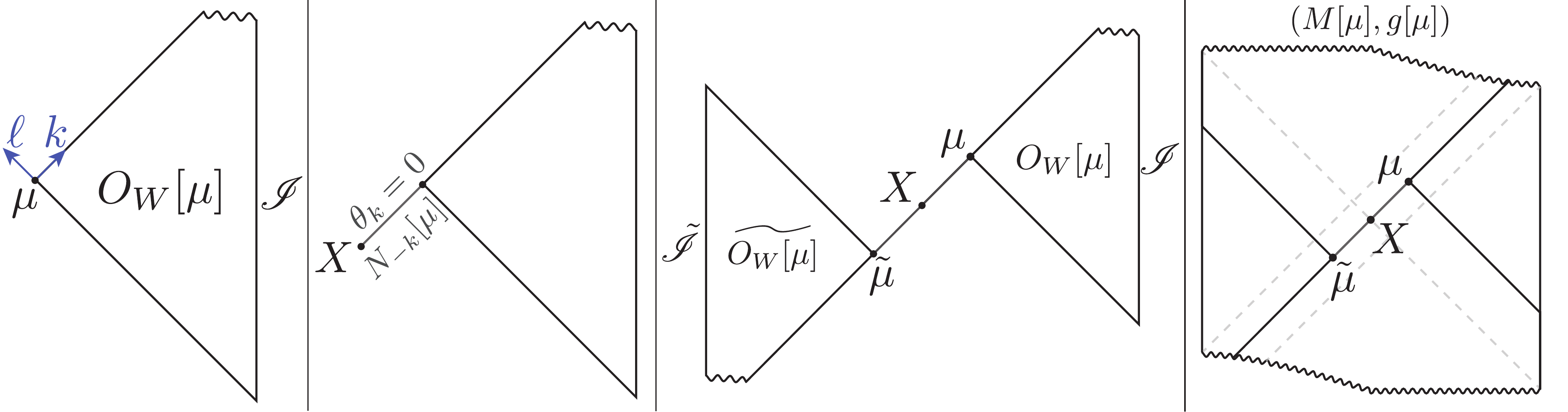}
\caption{The coarse-grained spacetime construction of~\cite{EngWal17b, EngWal18}. The first panel shows the outer wedge of a minimar $\mu$. The coarse grained spacetime elsewhere is
specified by a piecewise null stationary initial data hypersurface $N_{-k}[\mu]$ emanating in the $-k^{a}$ direction
from $\mu$. It is possible to prescribe initial data on $N_{-k}[\mu]$ so that there is an extremal $X$ surface on it. At this extremal surface, the spacetime is CPT reflected to generate a complete Cauchy slice of $(M[\mu],g[\mu])$. }
\label{fig:coarsegrain}
\end{figure}

For our purposes, the takeaway from this construction is that for any minimar surface, there
exists a spacetime $(M[\mu], g[\mu])$ in which $\mu$ is null-separated along $-k^{a}$ from the HRT surface $X$ of
$\mathscr{I}$. Furthermore, since the construction of $(M[\mu], g[\mu])$ makes no assumptions about global
hyperbolicity~\cite{EngWal18}, we are free to apply it in the context under consideration in this article. 

\section{Apparent Horizons Lie Behind Event Horizons\label{sec:StepOne}} 

In this section, we argue that minimar surfaces must lie behind event horizons. The proof hinges on the
well-known requirement discussed in Section~\ref{sec:intro} that for consistency of the proposal~\eqref{eq:HRT}, the HRT surface of $\mathscr{I}$ must be spacelike- or
null-separated to $\mathscr{I}$.\footnote{Readers concerned about recent results involving quantum extremal
surfaces~\cite{EngWal14} outside of horizons~\cite{AlmMah19} should defer their concern to Section~\ref{sec:disc}, where we discuss quantum corrections and non-standard boundary conditions. } In modifying the area of the HRT surface, we would by the holographic prescription also modify $S_{\mathrm{vN}}[\rho_{\mathrm{bdy}}]$. Since the latter is invariant under local unitary operations, we immediately arrive at a contradiction; thus HRT surfaces must lie behind event horizons~\cite{HeaHub14}. 

From this requirement of  causal wedge inclusion, we will show that minimar surfaces must also lie behind an event horizon in the original spacetime. To prove this second step, it is critical that strictly classical GR satisfying the null convergence condition admit no evaporating singularities. Let us now make this requirement precise:

\begin{defn}\label{def:evapsing} An asymptotically AdS spacetime $(M,g)$  is said to be devoid of evaporating singularities if
 for every closed set $K \subset M$, when $\partial \hat{K}$ is a compact hypersurface in the conformal completion then
    $\hat{K}$ is compact in the conformal completion. 
\end{defn}
This property, by Lemma 8.2.1 of~\cite{Wald}, ensures that no inextendible curves are imprisoned in
$\hat{K}$, so that as we follow an inextendible geodesic in $K$ they either leave $K$ through $\partial K$ or go to the conformal boundary. 
In particular, this rules out geodesic incompleteness between complete hypersurfaces that do not touch singularities. 
In the special case where $\partial K$ is the union of two achronal surfaces, then compactness of $\hat{K}$
implies global hyperbolicity of $\hat{K}$.

As explained in the introduction, we will assume that all \textit{strictly classical} solutions to GR that admit a UV completion in quantum gravity are devoid of evaporating singularities. 
We emphasize that this is a far weaker assumption than strong asymptotic predictability~\cite{Wald}, since spacetime can
violate global hyperbolicity arbitrarily badly near the asymptotic region owing to non-evaporating singularities.  Our
assumption further means that no permissible extension of the Cauchy horizon results in an evaporating singularity in classical GR. See Fig.~\ref{fig:nonevap} for some examples.

\begin{figure}
\centering
\includegraphics[width=1.0\textwidth]{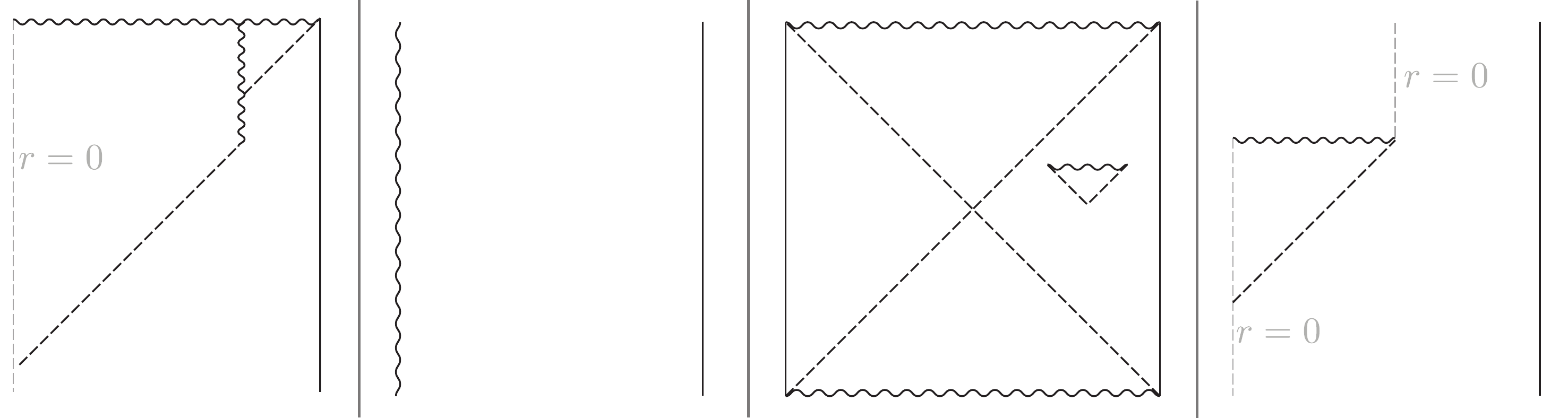}
\caption{Sketched conformal diagrams of four spacetimes. The two left-most spacetimes are neither globally hyperbolic nor strongly asymptotically predictable, but
    they are devoid of evaporating singularities. The two right-most spacetimes contain evaporating singularities.}
\label{fig:nonevap}
\end{figure}

From this, we prove that the choice of homology hypersurface is immaterial to the definition of $O_{W}[\mu]$:

\begin{lem} 
    \label{lem:H1H2}
    Let $(M,g)$ be devoid of evaporating singularities. If $H_{1}$ and $H_{2}$ are two homology hypersurfaces of a surface $\mu$ homologous to $\mathscr{I}$, then $D[H_{1}]=D[H_{2}]$. 
\end{lem}
\begin{proof}
   To prove this, it is sufficient to show that every inextendible timelike curve intersecting $H_1$ also either intersects $H_2$ or
    reaches the conformal boundary (since $H_{1}$ includes a slice of $\mathscr{I}$). Assume first that $H_{1} \cap H_{2}=\varnothing$, and let $C_1$ and $C_2$ be the
    Cauchy slices of $\mathscr{I}$ where they are anchored. By global hyperbolicity and spatial compactness of
    $\mathscr{I}$ there is a compact subset $ I \subset \mathscr{I}$ with $\partial I = C_1
    \cup C_2$. Then $\partial \hat{K}\equiv I \cup H_1 \cup H_2$ is a compact hypersurface in the conformal completion (see
    Fig.~\ref{fig:H1H2} for an illustration), and the region 
    $\hat{K}$ bounded by $\partial \hat{K}$ in the conformal completion is also compact. 

    Assume an inextendible timelike curve $\gamma$ enters $K$ through $H_1$, and assume without loss of generality
    that it enters $K$ to the future. 
    By compactness of $\hat{K}$, $\gamma$ must either reach the conformal boundary or leave $K$ to its future. 
    It cannot leave $K$ through $H_1$ by global hyperbolicity of $K$. 
    Thus if $\gamma$ does not go to the conformal boundary, it intersects $H_2$. 

    In the case where $H_1$ and $H_2$ intersect we potentially get multiple compact regions $\hat{K}_{1}, \hat{K}_{2}, \ldots$ bounded by 
    $H_1 \cup H_2 \cup I$, as shown in Fig.~\ref{fig:H1H2}. Since the above argument will apply to each region we again find that $\gamma$ intersects $H_2$
    or reaches $I$.\footnote{Here we are ignoring potential issues that could arise if the intersection is a dense measure zero set; we assume that the domain of dependence of any one of these hypersurfaces is codimension-zero, and thus that we can ``wiggle'' the hypersurface to avoid that scenario.}
\end{proof}

\begin{figure}
\centering
\includegraphics[width=0.6\textwidth]{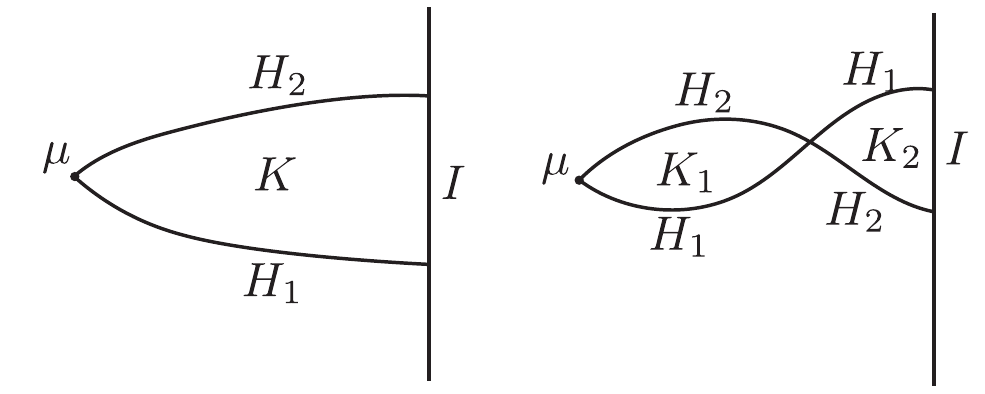}
\caption{Illustration of two homology slices, which by the lack of evaporating singularities, must have
the same outer wedge.}
\label{fig:H1H2}
\end{figure}
For the remaining of the paper, unless stated otherwise we will make our above assumption that classical GR allows no evaporating singularities that satisfy the null convergence condition.

\begin{thm} 
    \label{thm:muhide}
    Let $\mu$ be a future minimar surface in $(M,g)$. Then $\mu$ lies behind the future event horizon in $(M,g)$.
\end{thm}

\begin{proof} 
Let $\Sigma_{\mathrm{tot}}$ be an initial data slice of the spacetime $(M[\mu], g[\mu])$ consisting of $N_{-k}[\mu]$,
$H$, and $\widetilde{H}$.\footnote{Here $\widetilde{A}$ refers to the CPT conjugate of a quantity $A\in O_{W}[\mu]$ in
the coarse-grained spacetime.} Within $D[\Sigma_{\mathrm{tot}}]$, we know that the extremal surface $X[\mu]$
constructed via the above prescription is the minimal area extremal surface in $D[\Sigma_{\mathrm{tot}}]$. 
If this remains true for any permissible extension of the Cauchy horizon, then $X[\mu]$ must be the HRT surface of $(M[\mu], g[\mu])$. 

Suppose now that there is a permissible extension of the Cauchy horizon containing an extremal surface $X'$ also
homologous to $\mathscr{I}$. Since the total CFT state on all boundaries is pure, if $X'$ were to be the HRT surface of
$\mathscr{I}$,  then it would also be the HRT surface of the complement
$\widetilde{\mathscr{I}}$ by complementary recovery of classical holographic codes, and so it would have to be homologous to $\widetilde{\mathscr{I}}$. 
From the $\mathscr{I}$ and $\widetilde{\mathscr{I}}$ homology hypersurfaces of $X'$ we construct the complete hypersurface
$\Sigma_{\mathrm{tot}}'$ as the union (see illustration in Fig.~\ref{fig:Xprime}).
By the absence of evaporating singularities
and the requirement that homology hypersurfaces are compact in the conformal completion,
the region between $\Sigma_{\mathrm{tot}}$ and $\Sigma_{\mathrm{tot}}'$ is globally hyperbolic. 
But this means that $D[\Sigma_{\rm tot}'] = D[\Sigma_{\rm tot}]$, and so $X' \subset D[\Sigma_{\rm tot}]$. This
contradicts the assumption that $X'$ lies in the extension beyond the Cauchy horizon.
Thus the proof that $X[\mu]$ is the HRT surface from the case where $(M[\mu], g[\mu])$ is globally hyperbolic still
applies, and so $\mu$ must be behind the horizon in $(M[\mu], g[\mu])$.
\begin{figure}
\centering
\includegraphics[width=0.5\textwidth]{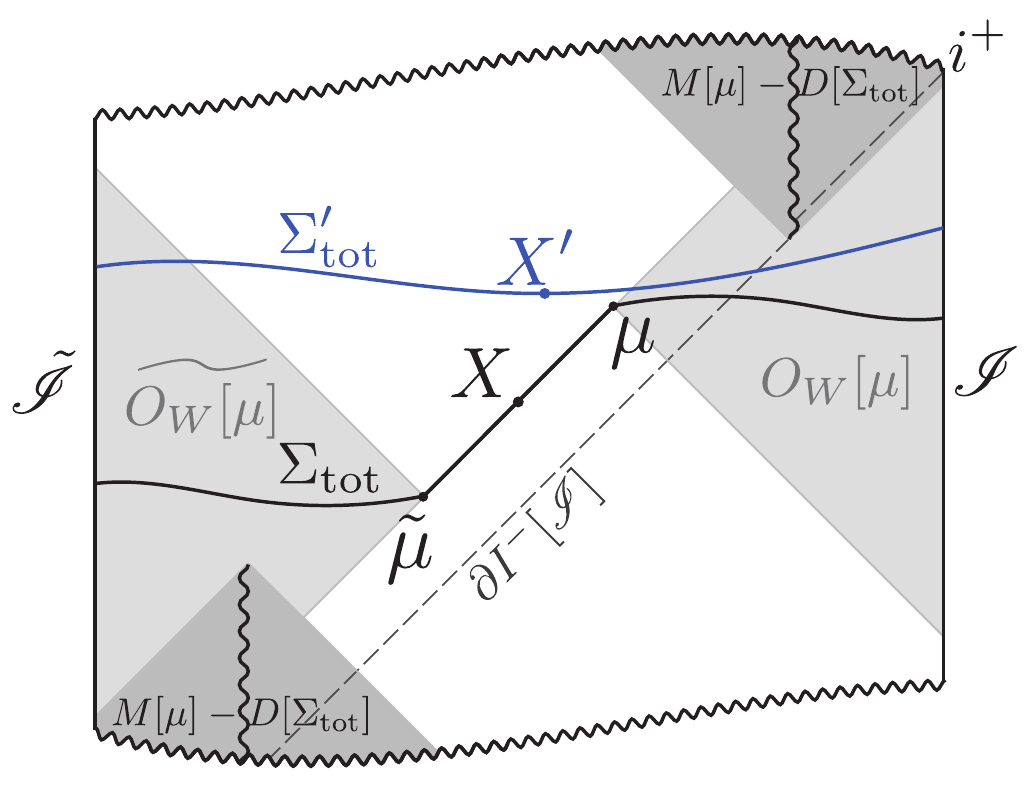}
\caption{Illustration of a coarse grained spacetime extended beyond Cauchy horizons into the dark gray regions. $X'$ is an extremal surface that is
    a candidate HRT surface. Since $X'$ must be homologous to both $\mathscr{I}$ and $\tilde{\mathscr{I}}$, it will lie
in $D[\Sigma_{\rm tot}]$. $\mu$ must be spacelike to $\tilde{i}^+$ for any choice of extension $M[\mu]-D[\Sigma_{\rm tot}]$.}
\label{fig:Xprime}
\end{figure}

Consider now the subset $Z = J^+[\partial^+ O_W[\mu]] - J^+[\mu]$  of $M[\mu]$, which we
can choose so that it agrees with the corresponding subset of $M$ (see Fig.~\ref{fig:wedgeExt} for an example). 
If $\widetilde{i^+} \not\subset Z$, then CFT evolution proceeds beyond $\mathscr{I}\cap Z$, and
$D[\Sigma_{\rm tot}] \cup Z$ can be extended further. However, a maximal extension 
of $D[\Sigma_{\rm tot}] \cup Z$ must put
$\mathscr{I}$ in
causal contact with $\mu$, violating causal wedge inclusion, and so we conclude that
$\widetilde{i^+} \subset Z$. Thus $\mu$ must lie behind the horizon in $O_W[\mu] \cup Z$, and by the
absence of closed timelike curves this must remain true when completing $O_W[\mu] \cup Z$ into
$M$.
\begin{figure}
\centering
\includegraphics[width=0.2\textwidth]{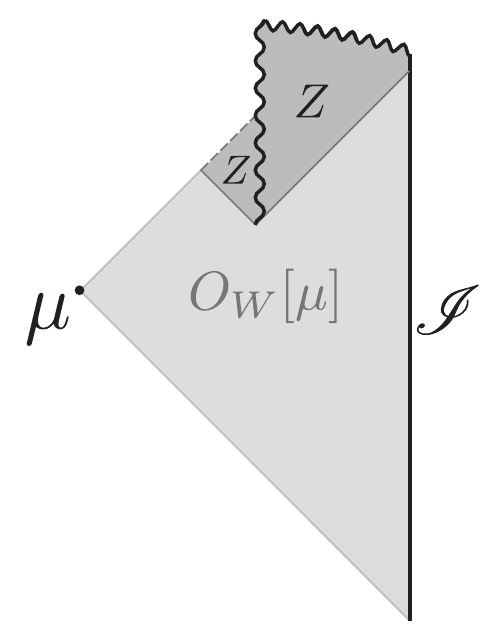}
\caption{Example of a Cauchy extension $Z$ of the outer wedge that is spacelike to $\mu$. The fine-grained
spacetime $M$ will induce a particular choice of $Z$. It is possible to choose the coarse 
grained spacetime so that $O_W[\mu]\cup Z \subset M[\mu]$.}
\label{fig:wedgeExt}
\end{figure}
\end{proof}

\section{Trapped Surfaces Lie Behind Event Horizons}
Having argued that minimar surfaces lie behind event horizons, we proceed to show the main result of this article, that
\textit{trapped} surfaces lie behind event horizons. The idea is simple: prove that between a trapped surface and
the asymptotic boundary there
always exists at least one minimar surface. We will make use of the following theorem:
\begin{thm}[\cite{AndMet07, AndEic10}]
    \label{thm:existence}
    Let $\Sigma$ be a compact spacelike hypersurface with an inner boundary and an outer boundary. Assume that the inner
    boundary is outer trapped and that the outer boundary is outer untrapped. Then $\Sigma$ contains a smooth stable marginally outer trapped surface $\sigma$. 
\end{thm}
\begin{figure}
\centering
\includegraphics[width=0.4\textwidth]{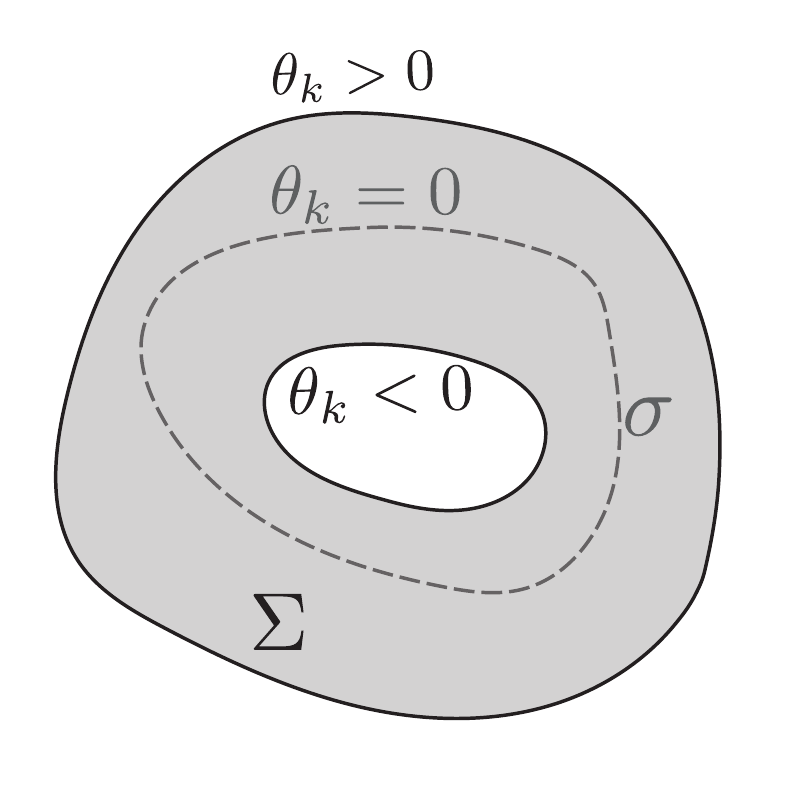}
\caption{Illustration of a compact spacelike hypersurface $\Sigma$ with an inner boundary that is outer trapped and an outer
boundary that is outer untrapped. By Theorem~\ref{thm:existence} a marginally outer trapped surface $\sigma$ in $\Sigma$ is
guaranteed to exist.}
\label{fig:MOTSthm}
\end{figure}
\noindent The quantities described in this theorem are illustrated in Fig.~\ref{fig:MOTSthm}. In fact, we will use a stronger version
of the above theorem that describes how $\sigma$ evolves with an evolving family of spatial slices. We will assume a generic condition on trapped surfaces which is frequently used in classical gravity proofs involved marginally trapped surfaces (see e.g.~\cite{AshKri03, BouEng15b}: $R_{ab}n^a n^b + \sigma_{ab}\sigma^{ab} > 0$, where $\sigma_{ab}$ is the shear tensor and $n^{a}$ is the generator of a null congruence fired orthogonally from the trapped surface. This condition ensures that every leaf in a spacelike
hypersurface foliated by marginally outer trapped surfaces is strictly stable~\cite{AshKri03}. Note that in the following we will only require this condition inside of the outer wedge. 

Combining now Theorem 2.1, 3.1 and 6.4 of \cite{AndMars08} with the assumption of genericity and its implication of
strict stability \cite{AshKri03}, we have the following:\footnote{The theorem is only proven for $D=4$ -- we will assume
it holds for $D\geq 3$, which appears likely given that Theorem~\ref{thm:existence} has been proven for $3\leq D \leq 8$ \cite{AndEic10}.} 
\begin{thm}[\cite{AndMars08}]\label{thm:existence2} Let $\Sigma$ be a spacelike hypersurface and let $\partial \Sigma$ consist of two disconnected components $\sigma_{1}$ and $\sigma_{2}$. Assume further that one of these components, $\sigma_{1}$ is outer trapped, while the other $\sigma_{2}$ is outer untrapped. Then the boundary of the outer trapped region on $\Sigma$, denoted $\sigma_{t}$, is (1) a smooth, strictly stable, marginally outer trapped surface homologous to  $\sigma_{2}$. Furthermore, if a spacetime region admits a foliation $\{\Sigma_{t}\}$ by such hypersurfaces, then the union $  \mathcal{H} = \bigcup_{t \in [0, T]} \sigma_{t}$ is a piecewise smooth spacelike manifold with a finite number of connected components. \end{thm}

To apply Theorem~\ref{thm:existence2} to prove that trapped surfaces lie behind event horizon, we need $\sigma_t$ to also be (1) marginally trapped and (2) homologous to $\mathscr{I}$. To achieve these
conditions, we require the implementation of a mild assumption about the past of trapped surfaces ($\theta_{\ell}<0$, $\theta_{k}<0$). 
Since cosmic censorship is primarily concerned with the future of trapped surfaces, these assumptions do
not heavily constrain our results.

\begin{defn}
    We call a surface $\sigma$  past well-behaved if it
    is homologous to $\mathscr{I}$ and (1)
    $C[\sigma]=\partial^{-}O_{W}[\sigma] \cap \mathscr{I}$
    is a Cauchy slice of $\mathscr{I}$ and (2)  $\partial^{-} O_{W}[\sigma] \subset D$ for some interior of a domain of dependence $D$.\footnote{
This is morally equivalent to demanding the existence of an open set $\mathcal{O}$ around $\sigma$ that can be covered by surfaces that also fulfill requirement (1).} 
\end{defn}

With this in place we prove a lemma essential for our main result:
\begin{lem}\label{lem:mainlem}
    Let $\tau$ be a past well-behaved trapped surface in a generic spacetime $(M, g)$. Then there exists 
    a spacelike manifold $H[\tau]$ in $O_W[\tau]$ which is foliated by
    past well-behaved strictly stable marginally outer trapped surfaces homologous $\mathscr{I}$.
\end{lem}
\begin{proof}
    Pick a smooth one-parameter family of non-intersecting spatial slices
    $\Sigma_{t}$ for $t\in[0, T]$ that are all contained in $D\cap O_W[\tau]$, where $D$ is the interior of the domain
    of dependence containing $\partial^{-}O_W[\tau]$. Pick the hypersurfaces $\Sigma_{t}$ so that they
    are anchored on Cauchy slices $C_t$ of the conformal boundary, and pick the inner boundaries of $\Sigma_{t}$ so they
    are either sufficiently close to $\tau$ or lying on $N_{k}[\tau]$, so that they are also trapped.
    By AdS asymptotics and global hyperbolicity in $D$, we can always choose the family so that each $\Sigma_t$ has an
    untrapped surface near $\mathscr{I}$. Theorem~\ref{thm:existence2} now guarantees the existence of a set $\mathcal{H}$ which is the union of smooth spacelike
    manifolds foliated by smooth strictly stable marginally outer trapped surfaces $\mu_{t}$ homologous to $C_{t}$ (and thus
    $\mathscr{I}$), together with surfaces where $\mathcal{H}$ jumps. See Fig.~\ref{fig:MOTT}. Since we choose our foliation so $\mathcal{H}\subset D$, 
    the $\mu_{t}$ are also past well-behaved.
    Hence, any of the connected components of $\mathcal{H}$, after removing jump surfaces, satisfies all of our claimed properties of $H[\mu]$.
    \begin{figure}
    \centering
    \includegraphics[width=0.4\textwidth]{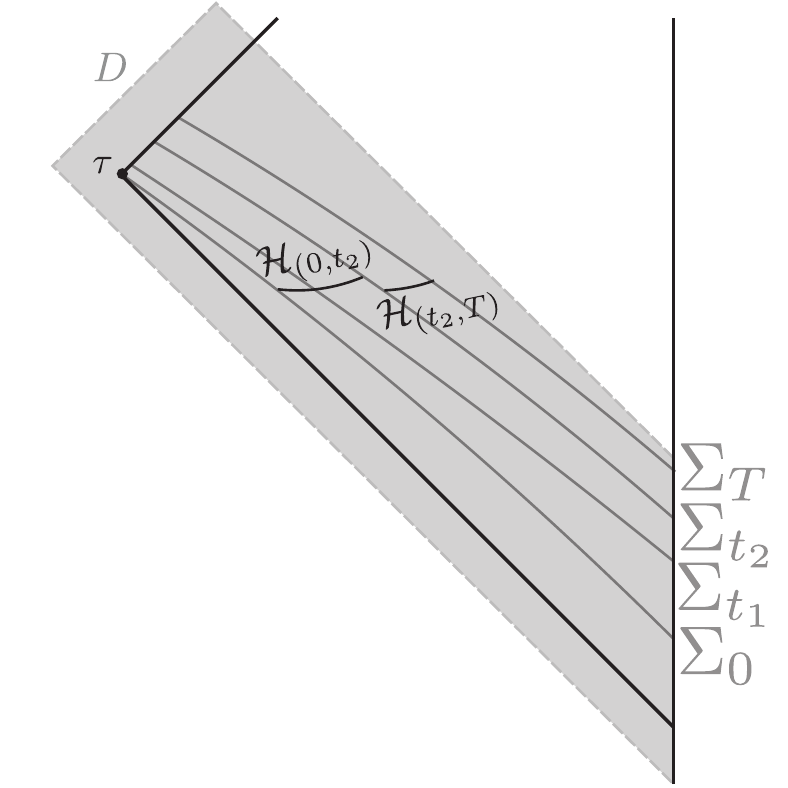}
        \caption{An example of the construction of the manifold $H[\tau]$ outside a past well-behaved trapped
        surface $\tau$. $\Sigma_{t}$ is a one-parameter family of spatial slices, and both $\mathcal{H}_{(0,
        t_2)}$ and $\mathcal{H}_{(t_2, T)}$ are spacelike manifolds foliated by marginally outer trapped surfaces.
        $H[\tau]$ can be taken to be either.}
    \label{fig:MOTT}
    \end{figure}
\end{proof}
    To find our main result, all that is missing is that at least one of the marginally outer trapped surfaces  from the previous lemma is also inner trapped. 
By the past well-behaved assumption, these surfaces cannot be inner untrapped (so-called ``anti-normal''). If the spacetime were spherically symmetric we would now be done: by past well-behavedness we
    can fire a spherically symmetric past horizon from $\mathscr{I}$ that hits $\mu_{t}$. Choosing $\Sigma_{t}$ (and thus
    $\mu_t$) to respect the spherical symmetry, $\mu_t$ would be contained in the horizon, and so
    $\theta_{\ell}[\mu_t]<0$, where strict inequality is guaranteed by our genericity condition. In the absence of spherical symmetry, this becomes harder to prove. We nevertheless find it to be a reasonable assumption for the following two reasons: (1) there is a continuously large amount of freedom in $\mu_t$ given that the choice of foliation $\Sigma_t$ is
    arbitrary. (2) By past well-behavedness we can fire a one-parameter family of horizons from the boundary in $D$ and pick
    each $\Sigma_t$ to lie arbitrarily close to a member of this family of horizons. Thus we can construct an $H[\mu]$ where every leaf
    lies arbitrarily close to a strictly inner trapped surface (strictness follows from genericity). 
    In fact, if we were allowed to use Theorem~\ref{thm:existence2} directly on null foliations with a measure zero set of non-differentiable points, then we
    could redo the proof of Lemma~\ref{lem:mainlem} with $\Sigma_{t}$ chosen to be future horizons, which would yield a proof that
    $\theta_{\ell}[\mu_t]<0$ as well. With these justifications in mind, we will simply assume that a past-well behaved $\tau$ has at
    least one choice of $H[\tau]$ containing an inner trapped leaf $\mu_t$. In light of the guaranteed existence of
    $H[\tau]$ by Lemma~\ref{lem:mainlem}, this can be considered as an addition to the definition of being past well-behaved.

Finally, we can show our main result. 
\begin{thm}
    \label{thm:mainthm}
    Let $\tau$ be a trapped surface. If $\tau$ is past well-behaved or has a trapped surface in its outer wedge which is past well behaved, then no future causal curves from $\tau$ can reach $\mathscr{I}$. In particular, $\tau$ lies behind an event horizon. 
\end{thm}
\begin{proof}
    By Lemma~\ref{lem:mainlem} we know there is a past well-behaved surface $\mu$ in $O_W[\tau]$ satisfying the
minimar conditions (1)--(3). We now show that for a surface satisfying past well-behavedness, only properties (1)--(3) of a minimar are needed for
the coarse grained spacetime to exist and having the required property that the extremal surface $X$ on $N_{-k}[\mu]$ is the HRT surface of $\mathscr{I}$. 
That is, the requirement of minimality of $\mu$ on a homology slice can be exchanged for past well-behavedness, and
Theorem~\ref{thm:muhide} still applies to $\mu$.

Existence of the coarse grained spacetime is immediate, since that only relies on properties (1)--(3), as shown in \cite{EngWal18}.
What remains is to show that $X$ is HRT. The crucial piece used for proving this in~\cite{EngWal18} was that $X$ is
contained in a Cauchy surface of the coarse grained spacetime in which it is minimal. 
We now show that there is such a Cauchy surface $S$ in the coarse grained spacetime even if we do not assume
minimality of $\mu$ on a homology slice. 

Let $\Sigma$ be the Cauchy surface of the coarse grained spacetime formed by union of the homology slice
of $\mu$ with respect to $\mathscr{I}$, its CPT conjugate, and $N_{-k}[\mu]$. Let us for now work in the maximal Cauchy development $D[\Sigma]$,
currently refraining from extending possible Cauchy horizons. 

By past well-behavedness in the original spacetime, $\partial^{-}O_{W}[\mu]$ is contained in a globally hyperbolic
set there, and there are no singularities in a small neighbourhood
around it. Since the coarse grained spacetime of $\mu$ shares the same data on $\partial^{-}O_{W}[\mu]$, any potential singularity in the coarse
grained spacetime has to be finitely separated from $\partial^- O_{W}[\mu]$.\footnote{We have not ruled out that there
is a stress tensor shock along $\partial^-O_W[\mu]$, but this would not ruin global hyperbolicity.}
Thus, $C[\mu]$ is contained in the interior of $D[\Sigma]$.
\begin{figure}
\centering
\includegraphics[width=0.5\textwidth]{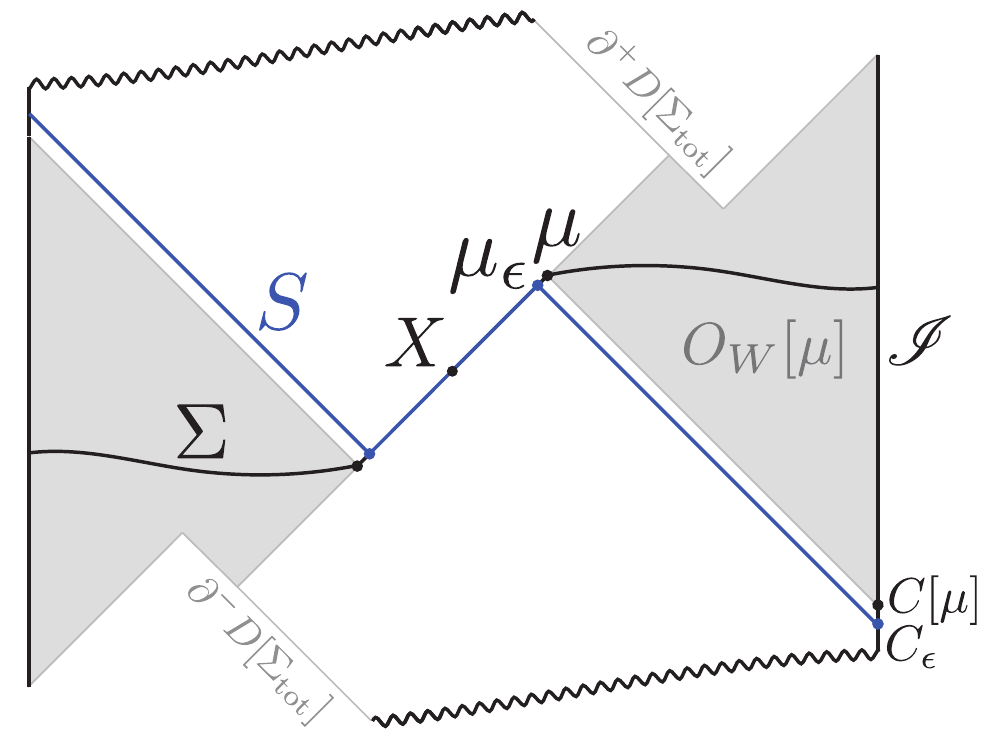}
\caption{Construction of Cauchy surface $S$ on which $X$ is minimal in the globally hyperbolic region of the coarse
grained spacetime $D[\Sigma]$.}
\label{fig:pseudomin}
\end{figure}
Generically, $C[\mu]$ will have kinks. 
Deform $C[\mu]$ slightly to the past into the Cauchy surface $C_{\epsilon}$ to smooth  out the kinks, and consider the past horizon $H = \partial I^+
[C_{\epsilon}] \cap J^-[\Sigma] $
and its intersection with the stationary null congruence, denoted $\mu_{\epsilon} = H \cap N_{-k}[\mu]$. See
Fig.~\ref{fig:pseudomin}. 
We know that if $\mu_{\epsilon}$ is nonempty, it must lie between $X$ and $\mu$ on $N_{-k}[\mu]$, since $X$ is not causally separated from
$\mathscr{I}$.
But since $\mu \subset I^+[C_{\epsilon}]$, we know that $\mu_{\epsilon}$ is nonempty and in fact a complete slice of $N_{-k}[\mu]$, since
otherwise a generator of $N_{-k}[\mu]$ going from $X$ to $\mu$ would be a causal curve that enters $I^+[C_{\epsilon}]$ without
intersecting $\partial I^{+}[C_{\epsilon}]$. 

We now take $S$ to be piecewise null Cauchy slice given by the union of $H$, $\tilde{H}$ and the part of
$N_{-k}[\mu]$ between $\mu_{\epsilon}$ and
its CPT conjugate. Since $H$ is a future horizon it has $\theta_{\ell}<0$, and since it ends on $N_{-k}[\mu]$, where
every slice has the same area as $\mu$, we get that any surface $\sigma$ contained in $H$ has $\mathrm{Area}[\sigma] \geq
\mathrm{Area}[\mu_{\epsilon}] = \mathrm{Area}[X]$. Clearly this also holds true for
$N_{-k}[\mu]$ and $\tilde{H}$, and thus for the whole of $S$.
Hence, we have constructed a Cauchy slice $S$ of $D[\Sigma]$ on which $X$ is minimal.
Even though $S$ does not intersect $O_W[\mu]$ or $\mu$, the proof from~\cite{EngWal18} that $X$ is the HRT
surface carries through. Furthermore, we are now free to add any permissible extension of the Cauchy horizon, and the
absence of evaporating singularities guarantees that $X$ remains that HRT surface, as shown in the proof of
Theorem~\ref{thm:muhide}. But now Theorem~\ref{thm:muhide} guarantees that $\mu$ lies behind the future event horizon
of $\mathscr{I}$, and since $\mu \subset O_W[\tau]$, $\tau$ lies behind the event horizon as well. 
\end{proof}

A simple corollary of our main result follows; we will relax our assumption that classical GR is devoid of evaporating singularities now:

\begin{cor}
    If there exists a past well behaved trapped surface $\tau$ in a classical asymptotically AdS spacetime $(M,g)$ satisfying the null convergence condition, then at least one of the following holds:
\begin{enumerate} 
	\item $(M,g)$ has an event horizon, and $\tau$ lies behind it;
	\item Classical GR admits solutions with evaporating singularities (in particular, there exist solutions with $O_{W}[\tau]$ and evaporating singularities);
    \item $(M, g)$ has no holographic dual.
	\end{enumerate}
\end{cor}

\section{Discussion}\label{sec:disc}
We have shown from the holographic dictionary that trapped surfaces, hallmarks of strong gravity, must be cloaked from the
asymptotic boundary by event horizons. This proof comes at the heels of the holographic derivation of the Penrose
Inequality in~\cite{EngHor19}; both results constitute strong evidence in favor of validity of a version of weak cosmic
censorship in gravitational theories whose UV completion is well-described by AdS/CFT, even though the
conjecture might be false in gravitational theories without such a UV completion. That is, these derivations suggest that there may be a formulation of weak cosmic censorship that is enforced by quantum gravity in the classical
regime (which may not be valid as a statement about gravitational theories that do not exist as classical
limits to quantum gravity). It is thus tempting to speculate that the location of trapped surfaces may serve as a
``swampland'' criterion of sorts~\cite{Vaf05}. We will not subscribe to such an interpretation here, but rather discuss the precise implications of our results in holography.

Before we do so, however, let us briefly comment on the assumption that evaporating singularities are absent in
classical GR. This may seem like an odd assumption to make given that the actual singularity region is not described by
the classical theory. However, the absence of of evaporating singularities is a constraint on the \textit{predictions}
of the classical theory: regardless of whether we expect quantum corrections to be important in that regime, we may
still ask what purely classical gravity would predict. This does not, however, preclude an investigation of the location
of \textit{quantum} trapped surfaces~\cite{Wal10QST} and quantum cosmic censorship in (perturbative) quantum gravity where the
null energy condition is violated; that is indeed a natural next step.\footnote{NE thanks R. Bousso and M. Tomasevic for extensive conversations on this.}

\paragraph{Typicality of Black Holes:}
The proof that trapped surfaces imply the existence of event horizons pairs nicely with the Penrose singularity theorem~\cite{Pen64}, which states that 
trapped surfaces imply the existence of singularities, and so singularities are at least as typical as trapped surfaces. For an asymptotic observer dreaming of immortality, the Penrose singularity theorem can be a foreboding
omen.\footnote{Indeed, as Wald describes it~\cite{Wal97}, the formation of naked singularities would be a potential
        mechanism for a mad scientist to destroy the universe.} 
However, our result should provide some relief to the observer, since it shows that \textit{event horizons are just as
typical}, in the sense that every time the Penrose theorem is invoked to deduce the existence of a singularity in a holographic theory, we 
can also deduce the existence of an event horizon hiding this singularity. 
Since a black hole usually is defined by the existence of an event horizon, our result shows that black holes are as typical as
trapped surfaces. Black holes always appear whenever gravity gets strong enough to focus both ingoing and outgoing
lightrays. This is especially reassuring given the prevalence of gedankenexperiments (in holography and beyond!) whose conclusions are reliant upon the formation of event horizons as a direct consequence of typical gravitational collapse.

\paragraph{Trapped Surfaces, Horizons, and Complexity:} In general, knowledge of the spacetime geometry at finite time is insufficient to locate event horizons: the latter are teleological with respect to $\mathscr{I}$. Our result
provides a time-local constraint on the existence and location of event horizons from the behavior
of trapped surfaces, whose location (and existence) can be ascertained in the neighborhood of a single moment of time. This
suggests a potential application to the work starting with~\cite{SheSta14, RobSta14, MalShe15} on the connection between the event horizon and maximal chaos and fast scrambling in holography. Should particular moment-of-time data in the CFT be sufficient to deduce the existence of a trapped surface, this data must be also be sufficient to detect such signatures of event horizons, although the converse is false. This raises
an interesting question: if trapped surfaces are a time-local guarantee for the existence of event horizons, what is the CFT dual of a trapped surface?

Some work has already been done in this direction: in~\cite{EngWal17b, EngWal18}, it was conjectured that the coarse-grained spacetime associated to an apparent horizon $\mu$
is dual to a CFT quantum state $\rho_{\rm coarse}$ that preserves all correlators of ``simple CFT operators'' with simple sources turned on:
\begin{equation}
\Tr[\rho\mathcal{O}_{\mathrm{simple}}] = \Tr\left[\rho_{\mathrm{coarse}}
\mathcal{O}_{\mathrm{simple}}\right],
\end{equation}
where by simple operators we mean CFT operators with support in $\mathscr{I}\cap O_W[\mu]$ that result in causal propagation in the bulk, and we have suppressed the sources (and $\rho$ is the original state). 
On the other hand, the expectation value of highly complex operators is allowed to differ between the two
states.
Consistent with expectations about complexity of operators localized to the deep black hole interior \cite{HarHay13,BroGha19}, operators localized in the spacetime behind $\mu$ should not be simply reconstructible with access to $\mathscr{I}\cap O_W[\mu]$.
In particular, the simple entropy proposal, if correct, would imply that correlation functions of a small number of the local CFT primaries of
HKLL ~\cite{HamKab05, HamKab06, HamKab06b} restricted to $O_{W}[\mu]\cap \mathscr{I}$ should not contain the information about the spacetime behind the associated minimar surface $\mu$. Our
results serve as an important consistency check to this: if a trapped surface could be in the causal wedge of
$\mathscr{I}$, then there would -- by Theorem 3 -- be a minimar surface in the causal wedge; but then parts of the region behind $\mu$ would be reconstructible from simple operators in $\mathscr{I}\cap O_{W}[\mu]$ by HKLL, in contradiction with the expectation that only complex operators are sensitive to the physics behind $\mu$. 

Thus, since trapped surfaces are generally behind minimar surfaces, trapped surfaces appear to be a robust time-local
signal of bulk physics that cannot be reconstructed from the CFT with ``simple experiments'' (under assumption of the simple
entropy proposal). In fact, via the simple entropy conjecture, there are protocols that employ simple operators that allow reconstruction of physics outside of minimars but inside
horizons (see also upcoming work~\cite{EngPenTA}). We could speculate that minimars (and their associated holographic screens~\cite{BouEng15a,BouEng15b}) are the boundary inside of which simple reconstruction methods break down; since we have shown that these surfaces are not found in the asymptotic region, our results are suggestive of a
time-local separation between simple and complex physics in the bulk not dissimilar from the one presented
in~\cite{BroGha19}.\footnote{The reader may prima facie suspect a contradiction here: our results apply to the nearest marginally trapped surface, whereas the Python's lunch conjecture applies to the nearest extremal surface. The difference lies in having access to the entire boundary as opposed to just $O_{W}[\mu]\cap\mathscr{I}$.This tension will be discussed more explicitly in~\cite{EngPenTA}.}

This is of course in harmony with the proposed relation between operator CFT complexity and bulk depth in \cite{Sus18, Sus19, SusZha20}. In
upcoming work~\cite{EngFol21}, we give a precise covariant description in which regions deeper in the black hole are of higher complexity than their shallower counterparts.

\paragraph{Quantum Corrections and Black Hole Evaporation:} In purely classical gravity, the assumption of the absence of evaporating singularities is quite reasonable: to our knowledge, known examples of evaporating singularities always involve some violation of the null energy
condition. Thus our statement should really be viewed as a prediction of quantum gravity on the behavior of strictly classical GR with matter that is consistent with holography. What about quantum corrections? 

Once quantum backreaction is included -- in particular, once black holes can evaporate -- it is possible for trapped surfaces to lie outside of event horizons. Interestingly, however, this effect appears at later times in the evaporation process (see e.g.\cite{BouEng15c}). Thus it is tempting to speculate that even in the perturbative quantum gravity regime, at times $t$ much smaller than the Page time, the classical GR result remains valid: trapped surfaces lie behind the event horizon. We emphasize that this is purely speculative: our derivation does not apply in such a setting since the event horizon is teleological and takes into account the entire evaporation process. 

Might we expect a quantum version of our statement to be valid? That is, do \textit{quantum} trapped
surfaces~\cite{Wal10QST} lie behind event horizons? In~\cite{AlmMah19}, quantum extremal surfaces~\cite{EngWal14} ``outside'' of the
horizon were found; in this case the existence of non-standard boundary conditions at the AdS boundary were crucial. Interestingly,  the quantum focusing conjecture~\cite{BouFis15a} in this case nevertheless enforces the absence of causal communication between the quantum extremal surface and $\mathscr{I}$. These complications illustrate the subtleties that must be accounted for in formulating a quantum version of our proof, in which evaporating singularities can no longer be ignored: the absence of causal communication from quantum trapped surfaces to $\mathscr{I}$ is not equivalent to the absence of quantum trapped surfaces in the causal wedge. This suggests that the correct generalization may actually involve an understanding of whether communication can occur \textit{in practice} rather than whether or not it is forbidden by causal structure.

\section*{Acknowledgments}
It is a pleasure to thank C. Akers, R. Bousso, R. Emparan, S. Fischetti, G. Horowitz, P. Jefferson, S. Leutheusser, and M. Tomasevic for discussions.  This work is supported in part by NSF grant no. PHY-2011905 and the MIT department of physics. The work of NE is also supported in part by the U.S. Department of Energy, Office of Science, Office of High Energy Physics of U.S. Department of Energy under grant Contract Number  DE-SC0012567 (High Energy Theory research). The work of \AA{}F is also
supported in part by an Aker Scholarship. 

\bibliographystyle{jhep}
\bibliography{all}

\providecommand{\href}[2]{#2}\begingroup\raggedright\begin{thebibliography}{10}

\bibitem{Haw71}
S.~W. Hawking, {\it Gravitational radiation from colliding black holes},  {\em
  Phys. Rev. Lett.} {\bf 26} (1971) 1344--1346.

\bibitem{Bek72}
J.~D. Bekenstein, {\it Black holes and the second law},  {\em Nuovo Cim. Lett.}
  {\bf 4} (1972) 737--740.

\bibitem{BarCar73}
J.~M. Bardeen, B.~Carter, and S.~W. Hawking, {\it The four laws of black hole
  mechanics},  {\em Commun. Math. Phys.} {\bf 31} (1973) 161.

\bibitem{FriSch93}
J.~L. Friedman, K.~Schleich, and D.~M. Witt, {\it {Topological censorship}},
  {\em Phys. Rev. Lett.} {\bf 71} (1993) 1486--1489,
  [\href{http://arxiv.org/abs/gr-qc/9305017}{{\tt gr-qc/9305017}}]. [Erratum:
  Phys.Rev.Lett. 75, 1872 (1995)].

\bibitem{Pen69}
R.~Penrose, {\it {Gravitational collapse: The role of general relativity}},
  {\em Riv. Nuovo Cim.} {\bf 1} (1969) 252--276. [Gen. Rel.
  Grav.34,1141(2002)].

\bibitem{GerHor79}
R.~P. Geroch and G.~T. Horowitz, {\it {Global structure of spacetimes}},  in
  {\em General Relativity: An Einstein Centenary Survey}, pp.~212--293.
\newblock 1979.

\bibitem{Wal97}
R.~M. Wald, {\it {Gravitational collapse and cosmic censorship}},  pp.~69--85,
  10, 1997.
\newblock \href{http://arxiv.org/abs/gr-qc/9710068}{{\tt gr-qc/9710068}}.

\bibitem{SantosTalk}
J.~Santos, {\it {Connecting the weak gravity conjecture to the weak cosmic
  censorship}}, . {Talk given at Strings 2018}.

\bibitem{HorSan16}
G.~T. Horowitz, J.~E. Santos, and B.~Way, {\it {Evidence for an Electrifying
  Violation of Cosmic Censorship}},  {\em Class. Quant. Grav.} {\bf 33} (2016),
  no.~19 195007, [\href{http://arxiv.org/abs/1604.06465}{{\tt
  arXiv:1604.06465}}].

\bibitem{CriSan16}
T.~Crisford and J.~E. Santos, {\it {Violating the Weak Cosmic Censorship
  Conjecture in Four-Dimensional Anti--de Sitter Space}},  {\em Phys. Rev.
  Lett.} {\bf 118} (2017), no.~18 181101,
  [\href{http://arxiv.org/abs/1702.05490}{{\tt arXiv:1702.05490}}].

\bibitem{CriHor17}
T.~Crisford, G.~T. Horowitz, and J.~E. Santos, {\it {Testing the Weak Gravity -
  Cosmic Censorship Connection}},  {\em Phys. Rev. D} {\bf 97} (2018), no.~6
  066005, [\href{http://arxiv.org/abs/1709.07880}{{\tt arXiv:1709.07880}}].

\bibitem{CriHor18}
T.~Crisford, G.~T. Horowitz, and J.~E. Santos, {\it {Attempts at vacuum
  counterexamples to cosmic censorship in AdS}},  {\em JHEP} {\bf 02} (2019)
  092, [\href{http://arxiv.org/abs/1805.06469}{{\tt arXiv:1805.06469}}].

\bibitem{HorSan19}
G.~T. Horowitz and J.~E. Santos, {\it {Further evidence for the weak gravity
  \textemdash{} cosmic censorship connection}},  {\em JHEP} {\bf 06} (2019)
  122, [\href{http://arxiv.org/abs/1901.11096}{{\tt arXiv:1901.11096}}].

\bibitem{GreLaf93}
R.~Gregory and R.~Laflamme, {\it {Black strings and p-branes are unstable}},
  {\em Phys. Rev. Lett.} {\bf 70} (1993) 2837--2840,
  [\href{http://arxiv.org/abs/hep-th/9301052}{{\tt hep-th/9301052}}].

\bibitem{GreLaf94}
R.~Gregory and R.~Laflamme, {\it {The Instability of charged black strings and
  p-branes}},  {\em Nucl. Phys. B} {\bf 428} (1994) 399--434,
  [\href{http://arxiv.org/abs/hep-th/9404071}{{\tt hep-th/9404071}}].

\bibitem{LehPre10}
L.~Lehner and F.~Pretorius, {\it {Black Strings, Low Viscosity Fluids, and
  Violation of Cosmic Censorship}},  {\em Phys. Rev. Lett.} {\bf 105} (2010)
  101102, [\href{http://arxiv.org/abs/1006.5960}{{\tt arXiv:1006.5960}}].

\bibitem{FigKun15}
P.~Figueras, M.~Kunesch, and S.~Tunyasuvunakool, {\it {End Point of Black Ring
  Instabilities and the Weak Cosmic Censorship Conjecture}},  {\em Phys. Rev.
  Lett.} {\bf 116} (2016), no.~7 071102,
  [\href{http://arxiv.org/abs/1512.04532}{{\tt arXiv:1512.04532}}].

\bibitem{FigKun17}
P.~Figueras, M.~Kunesch, L.~Lehner, and S.~Tunyasuvunakool, {\it {End Point of
  the Ultraspinning Instability and Violation of Cosmic Censorship}},  {\em
  Phys. Rev. Lett.} {\bf 118} (2017), no.~15 151103,
  [\href{http://arxiv.org/abs/1702.01755}{{\tt arXiv:1702.01755}}].

\bibitem{AndEmp18}
T.~Andrade, R.~Emparan, D.~Licht, and R.~Luna, {\it {Cosmic censorship
  violation in black hole collisions in higher dimensions}},  {\em JHEP} {\bf
  04} (2019) 121, [\href{http://arxiv.org/abs/1812.05017}{{\tt
  arXiv:1812.05017}}].

\bibitem{AndEmp19}
T.~Andrade, R.~Emparan, D.~Licht, and R.~Luna, {\it {Black hole collisions,
  instabilities, and cosmic censorship violation at large $D$}},  {\em JHEP}
  {\bf 09} (2019) 099, [\href{http://arxiv.org/abs/1908.03424}{{\tt
  arXiv:1908.03424}}].

\bibitem{AndFig20}
T.~Andrade, P.~Figueras, and U.~Sperhake, {\it {Violations of Weak Cosmic
  Censorship in Black Hole collisions}},
  \href{http://arxiv.org/abs/2011.03049}{{\tt arXiv:2011.03049}}.

\bibitem{Cho92}
M.~W. Choptuik, {\it {Universality and scaling in gravitational collapse of a
  massless scalar field}},  {\em Phys. Rev. Lett.} {\bf 70} (1993) 9--12.

\bibitem{Chr94}
D.~Christodoulou, {\it {Examples of naked singularity formation in the
  gravitational collapse of a scalar field}},  {\em Annals Math.} {\bf 140}
  (1994) 607--653.

\bibitem{Ham95}
R.~S. Hamade and J.~M. Stewart, {\it {The Spherically symmetric collapse of a
  massless scalar field}},  {\em Class. Quant. Grav.} {\bf 13} (1996) 497--512,
  [\href{http://arxiv.org/abs/gr-qc/9506044}{{\tt gr-qc/9506044}}].

\bibitem{HawEll}
S.~W. Hawking and G.~F.~R. Ellis, {\em The large scale stucture of space-time}.
\newblock Cambridge University Press, Cambridge, England, 1973.

\bibitem{Wald}
R.~M. Wald, {\em General Relativity}.
\newblock The University of Chicago Press, Chicago, 1984.

\bibitem{ArkMot06}
N.~Arkani-Hamed, L.~Motl, A.~Nicolis, and C.~Vafa, {\it The string landscape,
  black holes and gravity as the weakest force},  {\em JHEP} {\bf 06} (2007)
  060, [\href{http://arxiv.org/abs/hep-th/0601001}{{\tt hep-th/0601001}}].

\bibitem{Tip77}
F.~J. Tipler, {\it {Singularities in conformally flat spacetimes}},  {\em Phys.
  Lett. A} {\bf 64} (1977) 8--10.

\bibitem{Cla84}
C.~J.~S. {Clarke}, {\it {Naked Singularities and Causality Violation}},  in
  {\em Relativistic Astrophysics and Cosmology} (V.~{de Sabbata} and T.~M.
  {Karade}, eds.), vol.~1, p.~111, Jan., 1984.

\bibitem{Dra84}
C.~J.~S. {Clarke}, {\it {Naked Singularities and Causality Violation}},  in
  {\em Relativistic Astrophysics and Cosmology} (V.~{de Sabbata} and T.~M.
  {Karade}, eds.), vol.~1, p.~111, Jan., 1984.

\bibitem{Nol99}
B.~C. Nolan, {\it {Strengths of singularities in spherical symmetry}},  {\em
  Phys. Rev. D} {\bf 60} (1999) 024014,
  [\href{http://arxiv.org/abs/gr-qc/9902021}{{\tt gr-qc/9902021}}].

\bibitem{Ori00}
A.~Ori, {\it {Strength of curvature singularities}},  {\em Phys. Rev. D} {\bf
  61} (2000) 064016.

\bibitem{Emp20}
R.~Emparan, {\it {Predictivity lost, predictivity regained: a Miltonian cosmic
  censorship conjecture}},  {\em Int. J. Mod. Phys. D} {\bf 29} (2020), no.~14
  2043021, [\href{http://arxiv.org/abs/2005.07389}{{\tt arXiv:2005.07389}}].

\bibitem{EngHor19}
N.~Engelhardt and G.~T. Horowitz, {\it {Holographic argument for the Penrose
  inequality in AdS spacetimes}},  {\em Phys. Rev. D} {\bf 99} (2019), no.~12
  126009, [\href{http://arxiv.org/abs/1903.00555}{{\tt arXiv:1903.00555}}].

\bibitem{Pen73}
R.~{Penrose}, {\it {Naked Singularities}},  in {\em Sixth Texas Symposium on
  Relativistic Astrophysics} (D.~J. {Hegyi}, ed.), vol.~224 of {\em Annals of
  the New York Academy of Sciences}, p.~125, 1973.

\bibitem{ItkOz11}
I.~Itkin and Y.~Oz, {\it {Penrose Inequality for Asymptotically AdS Spaces}},
  {\em Phys. Lett. B} {\bf 708} (2012) 307--308,
  [\href{http://arxiv.org/abs/1106.2683}{{\tt arXiv:1106.2683}}].

\bibitem{HusSin17}
V.~Husain and S.~Singh, {\it {Penrose inequality in anti\textendash{}de Sitter
  space}},  {\em Phys. Rev. D} {\bf 96} (2017), no.~10 104055,
  [\href{http://arxiv.org/abs/1709.02395}{{\tt arXiv:1709.02395}}].

\bibitem{RyuTak06}
S.~Ryu and T.~Takayanagi, {\it {Holographic derivation of entanglement entropy
  from AdS/CFT}},  {\em Phys.Rev.Lett.} {\bf 96} (2006) 181602,
  [\href{http://arxiv.org/abs/hep-th/0603001}{{\tt hep-th/0603001}}].

\bibitem{HubRan07}
V.~E. Hubeny, M.~Rangamani, and T.~Takayanagi, {\it {A Covariant holographic
  entanglement entropy proposal}},  {\em JHEP} {\bf 0707} (2007) 062,
  [\href{http://arxiv.org/abs/0705.0016}{{\tt arXiv:0705.0016}}].

\bibitem{Wal12}
A.~C. Wall, {\it {Maximin Surfaces, and the Strong Subadditivity of the
  Covariant Holographic Entanglement Entropy}},  {\em Class.Quant.Grav.} {\bf
  31} (2014), no.~22 225007, [\href{http://arxiv.org/abs/1211.3494}{{\tt
  arXiv:1211.3494}}].

\bibitem{HeaHub14}
M.~Headrick, V.~E. Hubeny, A.~Lawrence, and M.~Rangamani, {\it {Causality \&
  holographic entanglement entropy}},  {\em JHEP} {\bf 12} (2014) 162,
  [\href{http://arxiv.org/abs/1408.6300}{{\tt arXiv:1408.6300}}].

\bibitem{BouEng15c}
R.~Bousso and N.~Engelhardt, {\it {Generalized Second Law for Cosmology}},
  {\em Phys. Rev.} {\bf D93} (2016), no.~2 024025,
  [\href{http://arxiv.org/abs/1510.02099}{{\tt arXiv:1510.02099}}].

\bibitem{AndMet07}
L.~Andersson and J.~Metzger, {\it {The Area of horizons and the trapped
  region}},  {\em Commun. Math. Phys.} {\bf 290} (2009) 941--972,
  [\href{http://arxiv.org/abs/0708.4252}{{\tt arXiv:0708.4252}}].

\bibitem{EngWal17b}
N.~Engelhardt and A.~C. Wall, {\it {Decoding the Apparent Horizon:
  Coarse-Grained Holographic Entropy}},  {\em Phys. Rev. Lett.} {\bf 121}
  (2018), no.~21 211301, [\href{http://arxiv.org/abs/1706.02038}{{\tt
  arXiv:1706.02038}}].

\bibitem{EngWal18}
N.~Engelhardt and A.~C. Wall, {\it {Coarse Graining Holographic Black Holes}},
  {\em JHEP} {\bf 05} (2019) 160, [\href{http://arxiv.org/abs/1806.01281}{{\tt
  arXiv:1806.01281}}].

\bibitem{AkeKoe16}
C.~Akers, J.~Koeller, S.~Leichenauer, and A.~Levine, {\it {Geometric
  Constraints from Subregion Duality Beyond the Classical Regime}},
  \href{http://arxiv.org/abs/1610.08968}{{\tt arXiv:1610.08968}}.

\bibitem{EngWal14}
N.~Engelhardt and A.~C. Wall, {\it {Quantum Extremal Surfaces: Holographic
  Entanglement Entropy beyond the Classical Regime}},  {\em JHEP} {\bf 01}
  (2015) 073, [\href{http://arxiv.org/abs/1408.3203}{{\tt arXiv:1408.3203}}].

\bibitem{CheWay19}
P.~M. Chesler and B.~Way, {\it {Holographic Signatures of Critical Collapse}},
  {\em Phys. Rev. Lett.} {\bf 122} (2019), no.~23 231101,
  [\href{http://arxiv.org/abs/1902.07218}{{\tt arXiv:1902.07218}}].

\bibitem{HerHor05}
T.~Hertog and G.~T. Horowitz, {\it {Holographic description of AdS
  cosmologies}},  {\em JHEP} {\bf 04} (2005) 005,
  [\href{http://arxiv.org/abs/hep-th/0503071}{{\tt hep-th/0503071}}].

\bibitem{AlmMah19}
A.~Almheiri, R.~Mahajan, and J.~Maldacena, {\it {Islands outside the horizon}},
   \href{http://arxiv.org/abs/1910.11077}{{\tt arXiv:1910.11077}}.

\bibitem{AndEic10}
L.~Andersson, M.~Eichmair, and J.~Metzger, {\it {Jang's equation and its
  applications to marginally trapped surfaces}},  in {\em {4th International
  Conference on Complex Analysis and Dynamical Systems}}, 6, 2010.
\newblock \href{http://arxiv.org/abs/1006.4601}{{\tt arXiv:1006.4601}}.

\bibitem{AshKri03}
A.~Ashtekar and B.~Krishnan, {\it {Dynamical horizons and their properties}},
  {\em Phys.Rev.} {\bf D68} (2003) 104030,
  [\href{http://arxiv.org/abs/gr-qc/0308033}{{\tt gr-qc/0308033}}].

\bibitem{BouEng15b}
R.~Bousso and N.~Engelhardt, {\it {Proof of a New Area Law in General
  Relativity}},  {\em Phys. Rev.} {\bf D92} (2015), no.~4 044031,
  [\href{http://arxiv.org/abs/1504.07660}{{\tt arXiv:1504.07660}}].

\bibitem{AndMars08}
L.~Andersson, M.~Mars, J.~Metzger, and W.~Simon, {\it The time evolution of
  marginally trapped surfaces},  {\em Classical and Quantum Gravity} {\bf 26}
  (Apr, 2009) 085018, [\href{http://arxiv.org/abs/0811.4721}{{\tt
  arXiv:0811.4721}}].

\bibitem{Vaf05}
C.~Vafa, {\it The string landscape and the swampland},
  \href{http://arxiv.org/abs/hep-th/0509212}{{\tt hep-th/0509212}}.

\bibitem{Wal10QST}
A.~C. Wall, {\it {The Generalized Second Law implies a Quantum Singularity
  Theorem}},  {\em Class.Quant.Grav.} {\bf 30} (2013) 165003,
  [\href{http://arxiv.org/abs/1010.5513}{{\tt arXiv:1010.5513}}].

\bibitem{Pen64}
R.~Penrose, {\it {Gravitational collapse and space-time singularities}},  {\em
  Phys. Rev. Lett.} {\bf 14} (1965) 57--59.

\bibitem{SheSta14}
S.~H. Shenker and D.~Stanford, {\it {Black holes and the butterfly effect}},
  {\em JHEP} {\bf 03} (2014) 067, [\href{http://arxiv.org/abs/1306.0622}{{\tt
  arXiv:1306.0622}}].

\bibitem{RobSta14}
D.~A. Roberts, D.~Stanford, and L.~Susskind, {\it {Localized shocks}},  {\em
  JHEP} {\bf 03} (2015) 051, [\href{http://arxiv.org/abs/1409.8180}{{\tt
  arXiv:1409.8180}}].

\bibitem{MalShe15}
J.~Maldacena, S.~H. Shenker, and D.~Stanford, {\it {A bound on chaos}},  {\em
  JHEP} {\bf 08} (2016) 106, [\href{http://arxiv.org/abs/1503.01409}{{\tt
  arXiv:1503.01409}}].

\bibitem{HarHay13}
D.~Harlow and P.~Hayden, {\it {Quantum Computation vs. Firewalls}},  {\em JHEP}
  {\bf 1306} (2013) 085, [\href{http://arxiv.org/abs/1301.4504}{{\tt
  arXiv:1301.4504}}].

\bibitem{BroGha19}
A.~R. Brown, H.~Gharibyan, G.~Penington, and L.~Susskind, {\it {The
  Python\textquoteright{}s Lunch: geometric obstructions to decoding Hawking
  radiation}},  {\em JHEP} {\bf 08} (2020) 121,
  [\href{http://arxiv.org/abs/1912.00228}{{\tt arXiv:1912.00228}}].

\bibitem{HamKab05}
A.~Hamilton, D.~N. Kabat, G.~Lifschytz, and D.~A. Lowe, {\it {Local bulk
  operators in AdS/CFT: A Boundary view of horizons and locality}},  {\em
  Phys.Rev.} {\bf D73} (2006) 086003,
  [\href{http://arxiv.org/abs/hep-th/0506118}{{\tt hep-th/0506118}}].

\bibitem{HamKab06}
A.~Hamilton, D.~N. Kabat, G.~Lifschytz, and D.~A. Lowe, {\it {Holographic
  representation of local bulk operators}},  {\em Phys.Rev.} {\bf D74} (2006)
  066009, [\href{http://arxiv.org/abs/hep-th/0606141}{{\tt hep-th/0606141}}].

\bibitem{HamKab06b}
A.~Hamilton, D.~N. Kabat, G.~Lifschytz, and D.~A. Lowe, {\it {Local bulk
  operators in AdS/CFT: A Holographic description of the black hole interior}},
   {\em Phys. Rev.} {\bf D75} (2007) 106001,
  [\href{http://arxiv.org/abs/hep-th/0612053}{{\tt hep-th/0612053}}]. [Erratum:
  Phys. Rev.D75,129902(2007)].

\bibitem{EngPenTA}
N.~Engelhardt, G.~Penington, and A.~Shahbazi-Moghaddam
  \href{http://arxiv.org/abs/Life without pythons would be so simple, to
  appear}{{\tt Life without pythons would be so simple, to appear}}.

\bibitem{BouEng15a}
R.~Bousso and N.~Engelhardt, {\it {New Area Law in General Relativity}},  {\em
  Phys. Rev. Lett.} {\bf 115} (2015), no.~8 081301,
  [\href{http://arxiv.org/abs/1504.07627}{{\tt arXiv:1504.07627}}].

\bibitem{Sus18}
L.~Susskind, {\it {Why do Things Fall?}},
  \href{http://arxiv.org/abs/1802.01198}{{\tt arXiv:1802.01198}}.

\bibitem{Sus19}
L.~Susskind, {\it {Complexity and Newton's Laws}},
  \href{http://arxiv.org/abs/1904.12819}{{\tt arXiv:1904.12819}}.

\bibitem{SusZha20}
L.~Susskind and Y.~Zhao, {\it {Complexity and Momentum}},
  \href{http://arxiv.org/abs/2006.03019}{{\tt arXiv:2006.03019}}.

\bibitem{EngFol21}
N.~Engelhardt and {\AA}.~Folkestad, ``to appear.''.

\bibitem{BouFis15a}
R.~Bousso, Z.~Fisher, S.~Leichenauer, and A.~C. Wall, {\it {Quantum focusing
  conjecture}},  {\em Phys. Rev.} {\bf D93} (2016), no.~6 064044,
  [\href{http://arxiv.org/abs/1506.02669}{{\tt arXiv:1506.02669}}].

\end{thebibliography}\endgroup

\end{document}